\documentclass[11pt]{article}
\usepackage{amssymb,amsfonts,amsmath,mathrsfs,nccmath,mathtools,geometry,authblk,dsfont,amsthm}
\usepackage{graphicx,xcolor,pstricks}
\usepackage[T1]{fontenc}
\usepackage{parskip}
\usepackage[T1]{fontenc}
\usepackage{parskip}
\usepackage{enumitem}
\usepackage[pdfborder={0 0 0}]{hyperref} 
\hypersetup{
colorlinks   = true,
citecolor    = blue,
linkcolor = blue
}

\numberwithin{equation}{section}
\makeatletter
\let\@fnsymbol\@zero
 \DeclareMathOperator{\sech}{sech}

\newcommand{\nvarphi} {\varphi^{\rm st}}

\newtheorem{theorem}{Theorem}[section]
\newtheorem{prop}[theorem]{Proposition}
\newtheorem{corollary}[theorem]{Corollary}
\newtheorem{problem}[theorem]{Problem}
\newtheorem{remark}[theorem]{Remark}

\begin{document} 

\title{The Schr\"odinger equation for the Rosen-Morse type potential revisited with applications}

\author[1]{Guillermo  Gordillo-Núñez\thanks{emails: \texttt{guillegordii@gmail.com, ran@us.es, niurka@us.es}}}
\affil[1]{Departamento de An\'alisis Matem\'atico, Universidad de Sevilla, c/Tarfia s/n, 41012, Sevilla, Spain}
\author[2]{Renato Alvarez-Nodarse}
\affil[2]{IMUS, Universidad de Sevilla, Departamento de An\'alisis Matem\'atico, Universidad de
Sevilla, c/Tarfia s/n, 41012, Sevilla, Spain}
\author[3]{Niurka R.\ Quintero}
\affil[3]{F\'\i sica Aplicada I, 
 Escuela T\'ecnica Superior de Ingenier\'\i a Inform\'atica,
Universidad de Sevilla, Avenida Reina Mercedes s/n, 41012 Sevilla, Spain}
\date{ \today}

\maketitle

\begin{abstract}
We rigorously solve the time-independent Schr\"odinger equation for the Rosen-Morse type potential. By using the Nikiforov-Uvarov method, we obtain, in a systematic way, the complete solution of such equation, which includes the so-called bound states (square-integrable solutions) associated with the discrete spectrum, as well as unbound states region (bounded but not necessarily square-integrable solutions) related to the continuous part of the spectrum.  The resolution of this problem is used to show that the  kinks of the non-linear Klein-Gordon equation with $\varphi^{2p+2}$ type potentials are stable. We also derive the orthogonality and completeness relations satisfied by the set of eigenfunctions which are useful in the description of the dynamics of kinks under perturbations or interacting with antikinks. 

\end{abstract}

\section{Introduction\label{intro}}  
 One-dimensional non-linear Klein-Gordon potentials, $U(\varphi)$, with multiple vacuum states support topological solitary waves \cite{fullin:1978,condat:1983,dashen:1974,lohe:1979,saxena:2019}, and even 
solitons in the case of the sine-Gordon potential \cite{rubinstein:1970}. More specifically, these  waves are solutions of the following non-linear Klein-Gordon equation, 
\begin{equation} \label{eq:nlkg}
	\varphi_{tt}-\varphi_{xx} + U'(\varphi)=0,
\end{equation}
where the subscripts $t$ and $x$ henceforth indicate the partial derivatives with respect to time and space, respectively, and  the prime denotes the derivative of the potential $U(\varphi)$ with respect to the field $\varphi$. 
In the case of kinks, these topological waves connect two minima of the potential, and  verify that
\begin{eqnarray}
	\label{sgbc2}
	\lim_{x \to + \infty}\varphi(x,t)=\lim_{x \to - \infty}\varphi(x,t) \pm Q, 
\end{eqnarray}
where the signs $\pm$ refer to the kink and anti-kink solution, respectively, and the constant $Q > 0$ is the so-called topological charge. Eq. \eqref{eq:nlkg} for different potentials $U(\varphi)$ and, consequently, kinks arising from them, has a vast list of applications (see the review by Campbell in Ref. \cite{campbell:2019}). For instance, the continuous second-order phase transition in the Ginzburg-Landau phenomenological theory can be explained through them \cite{landau:1950,ginzburg:2009,khare:2014}. Furthermore, kinks associated to the $\varphi^4$ type potentials are employed as mobile domain walls \cite{krumhansl:1975,koehler:1975,aubry:1976a}. In addition, the $\varphi^4$ kink with an excited internal mode leads to 
the fading of the  kink's wobbling, which can be explained by using a multiple scale expansion  \cite{barashenkov:2009,barashenkov:2019}. Nevertheless, the complete equations that model these kind of phenomena also contain additional terms related to external forces and damping. For these reasons, a study of the kinks' stability is necessary to establish observability of these particular solutions in real systems.

In particular, the stability of sine-Gordon and $\varphi^4$  kinks is determined by the resolution of certain Sturm-Liouville problem equivalent to the Schr\"odinger equation for the P\"oschl-Teller type potential \cite{poschl:1933}, $-l\,(l+1)\,\sech^2(x)$  for $l \in \mathbb{N}$  \cite{rubinstein:1970,dashen:1974,raban:2022},
which is symmetric and vanishes at the boundaries. In \cite{raban:2022}, it was found that the spectrum of the Schr\"odinger equation, composed of $l$ discrete eigenvalues and a continuous part, forms a complete set of eigenfunctions, and that these kinks are stable given that all the eigenvalues of the associated Sturm-Liouville problem are non negative.       

Actually, the knowledge of the spectrum is important not only because it determines the stability of the non-linear waves through non-negativity of its eigenvalues. In case of appearance of an internal mode, the loss of integrability is predicted, and, in addition, it is considered as a sufficient condition for inelastic collisions of solitary waves  \cite{bogdan:1990,kevrekidis:2001}. For example, when a kink and an antikink of a $\varphi^4$ system collide, they cannot pass through each other. Instead, depending on theirs initial velocities, they may form a bound state or be reflected, or, even, resonance windows may appear \cite{getmanov:1976,aubry:1976,makhankov:1978,sugiyama:1979}. 
These resonance windows, which are regions in which reflection and
trapping alternate until eventually the waves escape, were explained in terms of an exchange of energy  between translational and internal modes in the case of the $\varphi^4$ model \cite{campbell:1983}. 
Specifically, when the kink and the antikink meet, their internal modes are excited, and the energy oscillates between their respective translational and internal modes. Therefore, if appropriate resonance conditions are fulfilled, the energy can be restored to the translational modes and the kink and antikink simply move away from each other. This explanation was also supported by the resonance windows observed (or absent) in the kink-antikink collisions in the modified sine-Gordon equation with (or without) internal modes, respectively \cite{peyrard:1983}. 

However, resonance windows  were reported in the absence of internal modes in certain   
kink-antikink collisions for the $\varphi^6$ equation \cite{dorey:2011}. In this case, the non-linear Klein-Gordon potential has 3 absolute minima (vacua), with four exact solutions, two kinks and two antikinks \cite{khare:1979}. For the $\varphi^6$ equation, the Sturm-Liouville problem in whose resolution is based the linear stability analysis of each of these waves is no longer equivalent to the one-dimensional Schr\"odinger equation for the P\"oschl-Teller type potential. In fact, it is equivalent to the Schr\"odinger equation for the Rosen-Morse type potential \cite{rosen:1932,lohe:1979}
\begin{equation}
	v(z) =  v_{0} \cosh^{2}{\mu}\left(\tanh{z} + \tanh{\mu} \right)^{2},
	\label{rmpot}
\end{equation}
where $v_0>0$ and $\mu > 0$ are real parameters. This potential is also known as the Eckart potential \cite{eckart:1930} or as the Morse-Feshbach potential 
\cite{morse2:1953} and, in a way, constitutes a generalization of the P\"oschl-Teller potential as the latter can be obtained from \eqref{rmpot} by setting $\mu=0$ and changing the origin of energies. 

The Rosen-Morse potential is one of the solvable potentials \cite{cooper:1995} and is not symmetric, unlike the P\"oschl-Teller potential. Its resolution, in relation to the $\varphi^6$ system, shows that the discrete spectrum has only one mode, the so-called Goldstone mode (or the aforementioned translational mode) \cite{goldstone:1975}, i.e., there are no internal modes \cite{lohe:1979}. Nonetheless, as we advanced, the authors of \cite{dorey:2011} observed resonance windows in collisions between a kink and an antikink starting with an initial configuration of zero topological charge. Interestingly, this phenomenon was explained by the superposition of two Rosen-Morse potentials, one corresponding to the kink and the other to the antikink \cite{dorey:2011,roman:2019}. The shape of this effective potential is a well, where the meson bound states are formed, playing a similar role of that internal mode in the aforementioned sine-Gordon and $\varphi^4$ systems \cite{dorey:2011}. 
The meson bound states are described by the continuous spectrum of the Sturm-Liouville problem. They are also used to explain the effect of negative radiation pressure in $\varphi^4$ and $\varphi^6$ models, where the acceleration of the kinks towards the radiation was theoretically predicted and confirmed by simulations   \cite{forgacs:2008,roman:2019}. 
In the context of scalar field theory with polynomial self-interactions, the $\varphi^6$ potential has a positive mass-square unlike the  $\lambda \varphi^4$ model \cite{khare:1979}. An advantage of $\varphi^6$ and higher-order field theories  is that they can have more than two equilibrium state (vacua), such that, although $\varphi^6$ models a discontinuous first-order phase transition, $\varphi^8$, $\varphi^{10}$ and successive systems can describe successive phase transitions \cite{saxena:2019}.

Last but not least, regarding the Sturm-Liouville problem associated to the linear stability analysis, its full resolution is of importance as it provides an orthogonal and complete set of eigenfunctions, which can be used as a tool in the description of the dynamics of kinks under perturbations or interacting with antikinks \cite{barashenkov:2009,barashenkov:2019,kiselev:1998,sugiyama:1970}. More specifically, the unknown solitary wave solution of the perturbed equation may be expanded in this complete set of eigenfunctions \cite{bishop:1980,dauxois:2006}. Special significance is attached to
the Goldstone mode, since it determines the
motion of the soliton's center of mass.

The Schrödinger equation for the Rosen-Morse potential was introduced by Rosen and Morse  for describing the energy level structure of polyatomic molecules in Ref. \cite{rosen:1932},   where
after two successive and non-trivial changes of variables 
the bound states were obtained 
in terms of the Gauss's hypergeometric function. 
The most general study of this equation, 
to the best of our knowledge, is provided in  \cite[\S12.3]{morse2:1953} and most of the articles making use 
of its solutions (see for instance the Ref.\ \cite{lohe:1979}) refer to it. Nevertheless, it could be argued that the analysis of this equation 
carried out in Ref.  \cite{morse2:1953} is incomplete. 
To begin with, the transformations used to reduce the Schrödinger equation to a hypergeometric differential equation, which is discussed in the next section, 
are somewhat ``obscure'' whereas with the Nikiforov-Uvarov (NU) method they can be easily followed. 
Furthermore, when dealing with the bound states region, it is unknown whether the bound solutions obtained are the only 
square-integrable eigenfunction and, when dealing with the unbound states region, degeneracy of the continuous spectrum is not studied which leads to an incomplete solution of the problem. 
This \textit{incompleteness} makes the kink stability analysis previously mentioned not entirely exhaustive.

For these reasons, the aim of the current investigation is precisely to complete this study by using NU method \cite[\S1, page 1]{nu:1988}. 
The NU method, was developed by Nikiforov and Uvarov in the seventh decade of the XX century 
 in order to obtain the solutions of the differential equations  of hypergeometric type.
The idea is to transform an ordinary differential equation, the so-called generalized hypergeometric equation (see Eq.\ \eqref{GHE}), into 
the very well-known equation of hypergeometric type  (see Eq.\ \eqref{HDE}) by a viable change of variable in a constructive way.
 Among the solutions of the equations of hypergeometric type \eqref{HDE} are the classical orthogonal 
polynomials (Jacobi, Laguerre, Hermite) as well as the Gauss's hypergeometric functions.
In this way, the NU method constitutes a constructive approach for solving the  Schrödinger equation
for a wide class of potentials such as the potential well associated to the harmonic oscillator, 
the central attractive field  and the Coulomb potential, among others \cite[\S 1]{nu:1988}. 
In fact, this method was also used for finding the solution of the Schrödinger equation 
with a $q$-deformation of the Rosen-Morse potential in Ref. \cite{rezaei:2008} (see also references therein), where  
only the bound states were found.

The paper is organized as follows. 
In Sec. \ref{sec2}, by using the Nikiforov-Uvarov method, we fully solve the one-dimensional Schr\"odinger  equation  
with the Rosen-Morse potential, and obtain 
all eigenfunctions and eigenvalues corresponding not only to the bound states, but also to the continuous spectrum.  It is one of the
main goal of this paper as it 
enables the orthogonality and completeness of its  eigenfunctions to be proved 
 at the end of Sec. \ref{sec2},  
and, moreover, to explicitly obtain the so-called completeness relations. 
Section \ref{sec3} discusses some representative examples, 
including the linear stability analysis of the kinks of $\varphi^{2p+2}$ non-linear Klein-Gordon potentials, where the results of the previous section are used. 
  Sec. \ref{sec5} concludes  with a summary of our main results, and with a guide to solve the time-independent Schrödinger equation with the Rosen-Morse potential.

\section{The Schrödinger equation for the Rosen-Morse type potential \label{sec2}}

Let us begin by considering the time-independent Schrödinger equation
\begin{equation}
\left[\frac{d^2}{dz^{2}} - v(z) + \varepsilon \right] \psi(z) = 0, \quad z \in \mathbb{R}
\label{schrEqMFz}
\end{equation}
for the Rosen-Morse type potential $v(z)$ defined in Eq. \eqref{rmpot} with positive parameters $v_0$ and $\mu$. Here, $\varepsilon \in \mathbb{R}$ 
acts as the eigenenergy that needs to be found in order to solve 
the eigenvalue problem.

The cases involving a negative $\mu$ can be reduced to the previous ones, whereas whenever $v_0$ is negative the problem is 
completely different and we shall not consider it. In addition, the computations that we shall show can be particularized in a straightforward manner for the special case $\mu=0$, which leads to the so-called P\"oschl-Teller potential. Nevertheless, since there are some minor aspects to distinguish between the symmetric ($\mu = 0$) and asymmetric ($\mu \neq 0$) cases, we shall first deal with the latter and, 
 \textit{a posteriori}, we will show how the results obtained are adapted to the case $\mu=0$. 
With this choice of parameters ($v_0,\, \mu \in \mathbb{R}^{+} \setminus \{0\}$), the potential $v(z)$, resembles the diagram represented in Fig.\ \ref{fig:rmp}.
\begin{figure}[ht!]\centering
	\includegraphics[width=8cm]{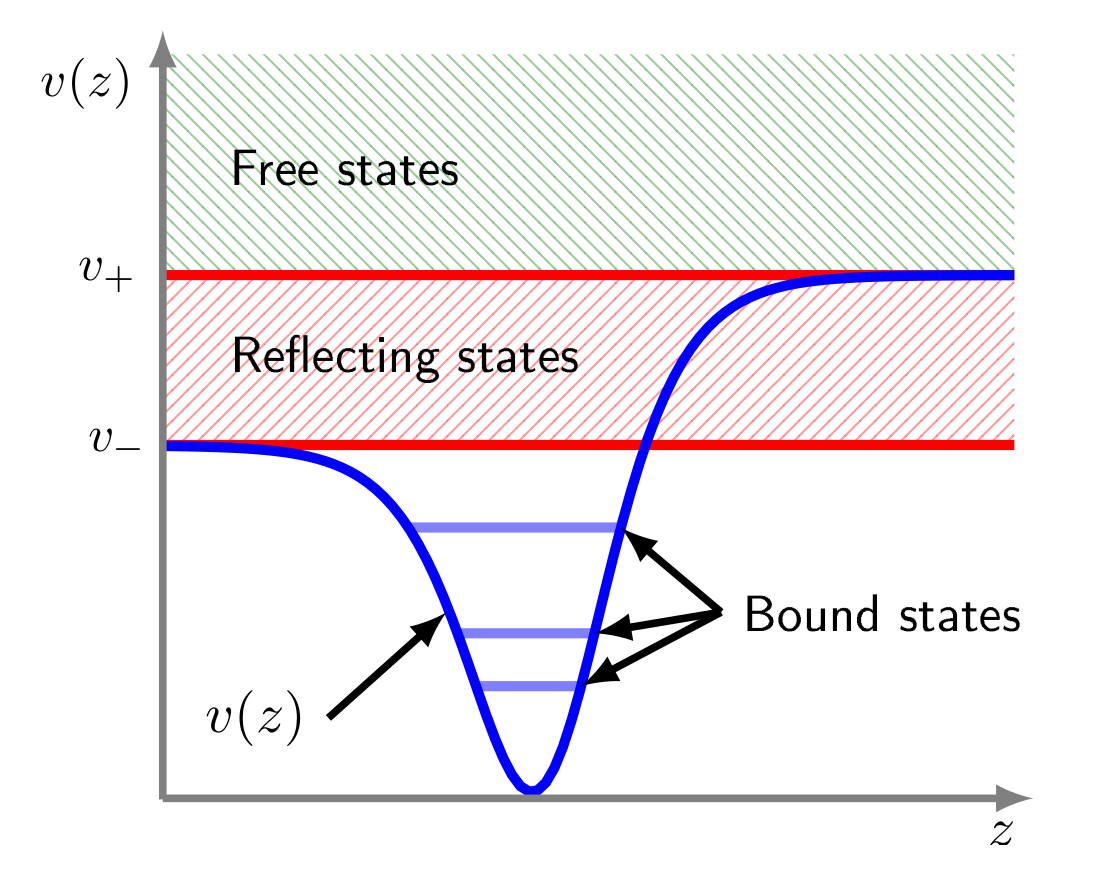}
\caption{The Rosen-Morse Potencial, $v(z) = v_{0} \cosh^{2}{\mu}\left[\tanh \left(z\right) + \tanh{\mu} \right]^{2}$, is represented for $\mu \ne 0$.  
The asymptotic values $v_{\pm}$ are given by Eq.\ \eqref{v+-}. 
 \label{fig:rmp}}
\end{figure}

Therefore, it is clear that the potential asymptotic behaviour is asymmetric  ($\mu \neq 0$) and, in this way, we may distinguish three different ranges of energy which shall result in three different states regions:
one bound and two unbound states regions (see Fig. \ref{fig:rmp}). In addition, these regions are delimited by the asymptotic values,
\begin{equation}
\label{v+-}
	v_\pm = \lim_{z \to \pm \infty} v(z) = v_{0}e^{\pm 2 \mu} \in \mathbb{R}^{+} \setminus \{0\},
\end{equation}
which satisfy that $0 < v_- < v_+$. 
Consequently, we define the three different energy states regions as follows:
\begin{enumerate}[label=\text{R\arabic*}]
\item \label{R1} For $\varepsilon \in (0,v_-)$ we distinguish the so-called \textit{bound states region}, where the particle would be classically confined in a finite region of space and 
the eigenvalue problem should have a nontrivial function, $\psi$, satisfying that $\psi \in L^{2}(\mathbb{R})$, this is, the space of square-integrable functions in 
$\mathbb{R}$.
\item \label{R2} For $\varepsilon  \in (v_-,v_+)$ we distinguish the so-called \textit{reflecting states region} which constitutes an unbound states region: the particle could (classically) reach $-\infty$ but not $+\infty$ and square-integrability condition on its eigenfunction could be dropped and substituted by $\psi \in L^{\infty}(\mathbb{R})$, this is, the space of bounded functions on $\mathbb{R}$. 
\item \label{R3} For $\varepsilon  \in (v_+,\infty)$ we distinguish the so-called \textit{free states region} which also constitutes an unbound states region: the particle could reach 
any point in the one-dimensional space and its eigenfunction would only need to be bounded, i.e., $\psi \in L^{\infty}(\mathbb{R})$.
\end{enumerate}
In addition, apart from the above three states regions, we need also to consider the three special states corresponding to $\varepsilon=0,v_{\pm}$.

The solutions of the Eq. \eqref{schrEqMFz} belonging to the first region \ref{R1} will be called \textit{bound solutions}, meaning the potential \textit{binds} the particle to a finite space, 
whereas the solutions for the last two regions \ref{R2} and \ref{R3} will be called \textit{unbound solutions}.

From the mathematical point of view we are interested in solving the following problem:
\begin{problem}
\label{schroProb}
Find all eigenvalues $\varepsilon \in \mathbb{R}$ for which the time-independent Schrödinger equation \eqref{schrEqMFz}
has non-trivial eigenfunctions, $\psi$, satisfying that 
\begin{itemize}
\item $\psi \in L^{2}(\mathbb{R})$ for the bound states region, that is, for $\varepsilon \in (0,v_-)$ and
\item $\psi \in L^{\infty}(\mathbb{R})$ for the unbound states regions, that is, for $\varepsilon \in (v_-, \infty)$.
\end{itemize}

\end{problem}
In other words, we will consider a Sturm-Liouville problem with boundary conditions given by $\psi \in L^{2}(\mathbb{R})$ or 
$\psi \in L^{\infty}(\mathbb{R})$ depending on the peculiarities of each region.

Before applying the NU method,  we set \(u = -\tanh{z}\) in Eq.\ \eqref{schrEqMFz} and, by dividing the obtained equation by $(1-u^2)^2$, it becomes into a generalized hypergeometric equation
 or GHE 
\begin{equation}
\psi''(u)+\frac{\widetilde{\tau}(u)}{\sigma(u)}\psi'(u)+\frac{\widetilde{\sigma}(u)}{\sigma^2(u)}\psi(u)=0, \quad u \in (-1,1),
\label{GHE}
\end{equation}
being $\sigma(u)$, $\Tilde{\sigma}(u) \in \mathbb{C}_2[u]$, and $\Tilde{\tau}(u)\in \mathbb{C}_1[u]$,
where $\mathbb{C}_p[u]$ denotes the set of polynomials with complex coefficients and degree at most $p$
on the variable $u$. In particular, for the computations described,
\begin{equation}
\widetilde{\tau}(u)=-2u, \quad \sigma(u) = 1-u^{2}, \quad  \widetilde{\sigma}(u;\varepsilon) = \varepsilon-v_{0}\cosh^{2}{\mu}\left(u - \tanh{\mu}\right)^{2}.
\label{gheMFPar}
\end{equation}
The idea behind the NU method can be summarize as follows: to find the solution of \eqref{GHE} in the form 
\begin{equation}
\label{anz}
\psi(u)=\phi(u)y(u),
\end{equation}
for a convenient function $\phi(u)$ in such a way that the obtained equation for $y(u)$ is no more complicated than 
the original GHE \eqref{GHE}. In fact, by choosing a $\phi(u)$ fulfilling
\begin{equation}
\frac{\phi'(u)}{\phi(u)} = \frac{\pi(u)}{\sigma(u)}, 
\label{phi}
\end{equation}
for an arbitrary $\pi(u) \in \mathbb{C}_1[u]$, function $y(u)$ satisfies another GHE of the form
\begin{equation} y''(u)+\frac{\tau(u)}{\sigma(u)}y'(u)+\frac{\overline{\sigma}(u)}{\sigma^2(u)}y(u)=0,
\label{GHEy}
\end{equation}
where
\begin{eqnarray}
\tau(u)&=&\widetilde{\tau}(u)+2\pi(u), \label{tau} \\
\overline{\sigma}(u)&=&\widetilde{\sigma}(u)+\pi^2(u)+
\pi(u)[\widetilde{\tau}(u)-\sigma'(u)]+\pi'(u)\sigma(u).
\label{sigmabar}
\end{eqnarray}
Notice that $\tau(u) \in \mathbb{C}_1$ and $\overline{\sigma}(u) \in \mathbb{C}_2[u]$ since we have set $\pi(u) \in \mathbb{C}_1[u]$.
Furthermore, if we choose a proper $\pi(u)$, which is equivalent to choosing 
$\phi(u)$ or $\tau(u)$ accordingly to Eqs.\ \eqref{phi} and \eqref{tau}, respectively, then $\overline{\sigma}(u)$ can be forced to be proportional to $\sigma(u)$, that is, 
$\overline{\sigma}(u) = \lambda\sigma(u),$ for an unkown $\lambda \in \mathbb{C}$. In this way, Eq.\ \eqref{GHEy} is reduced to
\begin{equation} \label{HDE}
\sigma(u) y''(u) + \tau(u) y'(u) +\lambda y(u)=0,
\end{equation}
which, indeed, has a simpler form than Eq.\ \eqref{GHE}. Equation \eqref{HDE}, whenever $\sigma(u) \in \mathbb{C}_2[z]$, $\tau(u)\in \mathbb{C}_1[z]$ and $\lambda \in \mathbb{C}$, 
is said to be a hypergeometric differential equation or HDE. 

Next, in order to make this result actually useful, we shall show a way to obtain a proper $\pi(u) \in \mathbb{C}_1[u]$ so that Eq.\ \eqref{GHEy} is, indeed, transformed into Eq.\ \eqref{HDE}. First of all, imposing condition $\overline{\sigma}(u) = \lambda\sigma(u)$ on Eq. \eqref{sigmabar} yields the following equation for the polynomial $\pi(u)$ and the unknown $k \in \mathbb{C}$,
\begin{equation}\label{eq-pi} 
\pi^2(u)+[\widetilde{\tau}(u)-\sigma'(u)]\pi(u)+ [\widetilde{\sigma}(u)-k\sigma(u)]=0,
\end{equation}
where, notice, 
\begin{equation}\label{eq-la} 
\lambda=k+\pi'(u), 
\end{equation}
that is, $\lambda$ is determined through $k$ and $\pi(u)$. 

In order to find $\pi(u)$ and $k$, for the sake of simplicity, we deviate from the standard procedure explained 
in \cite[\S1, page 1]{nu:1988}.  
Since $\pi(u)\in\mathbb{C}_1[u]$, we can write it as $\pi(u)=a-b\,u$, where $a, b \in \mathbb{C}$ are two coefficients to be determined. By inserting it in the above Eq. \eqref{eq-pi} and taking into
account Eq. \eqref{gheMFPar}, we obtain the following system of equations
\begin{equation}
\label{sys-abk}
a^2-k+\varepsilon-v_0\sinh^2\mu=0,\quad 2ab-v_0\sinh(2\mu)=0,\quad b^2+k-v_0\cosh^2\mu=0,
\end{equation}
that has four possible solutions, $(a_{j},b_{j},k_{j}(\epsilon))$ for $j=1, 2, 3, 4$, which define four pairs of $(\pi_j(u;\varepsilon), k_j(\varepsilon))$ as follows
\begin{equation}
\label{sol-pi}
\begin{aligned}
\pi_1(u;\varepsilon)&=\frac{\varkappa_+(\varepsilon) - \varkappa_-(\varepsilon)}{2}-\frac{\varkappa_+(\varepsilon) + \varkappa_-(\varepsilon)}{2}u,  && k_1(\varepsilon)= \frac{\varepsilon+v_0}{2}-\frac{\varkappa_+(\varepsilon) \varkappa_-(\varepsilon)}{2},\\
\pi_2(u;\varepsilon)&=-\pi_1(u;\varepsilon), && k_2(\varepsilon)=k_1(\varepsilon),\\ 
\pi_3(u;\varepsilon)&=\frac{\varkappa_+(\varepsilon) + \varkappa_-(\varepsilon)}{2}-\frac{\varkappa_+(\varepsilon) - \varkappa_-(\varepsilon)}{2}u,  && k_3(\varepsilon)= \frac{\varepsilon+v_0}{2}+\frac{\varkappa_+(\varepsilon) \varkappa_-(\varepsilon)}{2},\\
\pi_4(u;\varepsilon)&=-\pi_3(u;\varepsilon), && k_4(\varepsilon)=k_3(\varepsilon),\\ 
\end{aligned}
\end{equation}
where 
\begin{equation}\label{varkappa}
\varkappa_+(\varepsilon)=\sqrt{v_{+}-\varepsilon},\quad \varkappa_-(\varepsilon)= \sqrt{v_{-}-\varepsilon}.
\end{equation}
These four solutions are a direct consequence of two symmetries: (i) Eq.\ \eqref{eq-pi}  is invariant under the transformation $\pi \mapsto -\pi$ for a fixed $k$, since $\tilde{\tau}(u)=\sigma'(u)$ accordingly to Eq. \eqref{gheMFPar}; and
(ii) the system of equations \eqref{sys-abk} is also invariant if we interchange $a$ and $b$ and simultaneously replace $k$ by $-k+v_0+\epsilon$. Thus, any solution of \eqref{sol-pi} can be obtained by employing these symmetries into one of them. 

In addition, we can make use of Eq.\ \eqref{phi} to find that
\begin{equation}
\phi(u;\varepsilon)  = \left(1-u\right)^{\frac{\varkappa_-(\varepsilon)}{2}}\left(1+u\right)^{\frac{\varkappa_+(\varepsilon)}{2}}.
\label{rhophiMF}
\end{equation}
Moreover, from Eqs. \eqref{gheMFPar}, \eqref{tau} and \eqref{eq-la} we obtain the pairs $(\tau_j(u;\varepsilon),\lambda_j(\varepsilon))$ as follows
\begin{equation}
\tau_{j}(u;\varepsilon) = -2u +2\pi_j(u), \quad \lambda_j(\varepsilon) = k_{j}(\varepsilon) +\pi_j'(u,\varepsilon),\quad j=1,2,3,4.
\label{taulambda}\
\end{equation}
In such manner, we have obtained four different HDE \eqref{HDE} whose resolution is equivalent to that of the original GHE \eqref{GHE} and,  therefore, our next step is to make an appropriate choice of the HDE to solve attending to the peculiarities of each energy states region.

\subsection{Bound states region, $\varepsilon \in (0,v_-)$}

For the bound states region, $\varepsilon \in (0,v_-)$,  the solutions of Eq. \eqref{schrEqMFz} should be of integrable square 
(recall shape of Problem \ref{schroProb}). In order to obtain all and only functions of this kind, we shall make use of Theorem \cite[page 67]{nu:1988} and, to that purpose, the Remark in \cite[page 67]{nu:1988}
(see Theorem \ref{eigenTheo} and Proposition \ref{easesChoiceProp}, respectively) points us towards choosing a $\tau(u)$ satisfying that it vanishes at some point of $(-1,1)$ and whose derivative $\tau'(u) < 0$. 

Notice that, from Eq. \eqref{taulambda}, $\tau_2$ and $\tau_4$ do not satisfy 
the condition $\tau'(u)<0$ for any positive value of $v_0$ and $\mu$. Moreover, straightforward calculations show that only 
$\tau_1$ has a zero inside the interval $(-1,1)$ for all values of $\varepsilon \in (0,v_-)$, and every $v_0>0$ and $\mu>0$. For these very reasons, we choose to solve the HDE \eqref{HDE} associated to the pair $(\tau_1(u;\varepsilon),\lambda_1(\varepsilon))$.

Therefore, we look for solutions, $y(u;\varepsilon)$, of the following HDE
\begin{equation}
\left(1-u^{2}\right)y''(u)+2\left[a(\varepsilon)
-\left(b(\varepsilon)+1\right)u\right]y'(u)+
\left(k(\varepsilon) - b(\varepsilon)\right)y(u) = 0,
\label{hdeMF}
\end{equation}
where $u \in (-1,1)$ and  $k(\varepsilon)$, $a(\varepsilon)$, and $b(\varepsilon)$, are given by
\begin{equation}\label{R1ab}
k(\varepsilon)=\frac{\varepsilon+v_0}{2}-\frac{\varkappa_+(\varepsilon) \varkappa_-(\varepsilon)}{2},\quad
a(\varepsilon)=\frac{\varkappa_+(\varepsilon) - \varkappa_-(\varepsilon)}{2},\quad b(\varepsilon)=\frac{\varkappa_+(\varepsilon) + \varkappa_-(\varepsilon)}{2},
\end{equation}
respectively. 
Accordingly to above analysis,  the possible solutions of Eq.\ \eqref{schrEqMFz}, $\psi_{\varepsilon}(z)$, can be written in the form $\psi_{\varepsilon}(z) \propto \phi(-\tanh{z};\varepsilon)y(-\tanh{z},\varepsilon)$, where $\phi(-\tanh{z};\varepsilon)$ is given by Eq. \eqref{rhophiMF}.

Thus, the Nikiforov-Uvarov Theorem \cite[page 67]{nu:1988} (see Theorem \ref{eigenTheo} in the Appendix \ref{apen}) leads to the following result:

\begin{theorem}
\label{theoBounded}
The only non-trivial, bound solutions of time-independent Schrödinger equation \eqref{schrEqMFz}, $(\varepsilon,\psi_{\varepsilon}(z))$, for $\varepsilon \in (0,v_-)$ satisfying that 
$\psi_{\varepsilon}(z) \in L^2(\mathbb{R})$, that is, the solutions of Problem \ref{schroProb} for region \ref{R1}, are
\begin{equation}
\psi_{n}(z) = \mathcal{N}_n e^{-a_n z} \sech^{b_n}(z) P_{n}^{(b_n-a_n,b_n+a_n)}\left(-\tanh{z}\right),
\label{eigenfunctionsSchroMF}
\end{equation}
where 
$\mathcal{N}_n \in \mathbb{R}$ is a normalizing constant 
such that $\int_\mathbb{R} | \psi_{\varepsilon_n}(z)|^2dx=1$ ($\mathcal{N}_n$ is computed in Appendix \ref{comp-N_n}), 
and their corresponding eigenvalues
$\varepsilon = \varepsilon_n \in (0,v_-)$ fulfill the eigenenergy equation
\begin{equation}
k_n-b_n = n(n+2b_n+1), \quad n \in \mathbb{N} \cup \{0\},
\label{eigenvalueEqMF}
\end{equation}
which has solutions for a finite number of integers $n$ (see Remark \ref{finitenessOfBoundSolutionsRemark} for further explanations), where $a_n \coloneqq a(\varepsilon_n)$, $b_n \coloneqq b(\varepsilon_n)$ and $k_n \coloneqq k(\varepsilon_n)$. Here, $P_{n}^{(b_n-a_n,b_n+a_n)}(x)$ stands for the classical Jacobi polynomials (see Proposition \ref{jac-pol}).
\end{theorem}
\begin{proof} We use Theorem \ref{eigenTheo}. Since $\rho(u;\varepsilon)$ satisfies the Pearson equation \eqref{rho}, it can be shown that $\sqrt{\rho(u;\varepsilon)}$ satisfies the Eq.\ \eqref{phi}. Therefore, 
$\rho(u;\varepsilon) =\phi^2(u;\varepsilon)$, meaning
\begin{equation} \label{eqrho}
\rho(u;\varepsilon) = \left(1-u\right)^{\varkappa_-(\varepsilon)}\left(1+u\right)^{\varkappa_+(\varepsilon)}. 
 \end{equation}
Moreover, it can be checked that $\rho(u;\varepsilon)$ is bounded and fulfills condition \eqref{orthoPol} of classical orthogonal polynomials, that is, for every fixed $\varepsilon \in (0,v_-)$, it satisfies that 
\begin{enumerate}
\item $\rho(u;\varepsilon) \in L^{\infty}(-1,1)$,
\item $  u^{k} \sigma(u) \rho(u;\varepsilon)\Bigr|_{u=-1,1} = u^{k} \left(1-u\right)^{\varkappa_-(\varepsilon)+1}\left(1+u\right)^{\varkappa_+(\varepsilon)+1}\Bigr|_{u=-1,1}= 0, \quad \forall k \in \mathbb{N} \cup \{0\}$,
\end{enumerate}
which immediately follows from Eq. \eqref{varkappa} ($\varkappa_\pm>0$). Consequently, by virtue of Theorem \ref{eigenTheo} non-trivial solutions of Eq. \eqref{hdeMF}  satisfying that $y(u;\varepsilon)\sqrt{\rho(u;\varepsilon)}$ belongs to $L^{2}(-1,1) \cap L^{\infty}(-1,1)$ exist only when $\lambda = \lambda_n$, $n \in \mathbb{N} \cup \{0\}$, where $\lambda_n$ is given by \eqref{eq-lambda} and translates to the relation \eqref{eigenvalueEqMF}, and have the form $$y_n(u)=\dfrac{B_n}{\rho(u;\varepsilon)}\left[\sigma^n(u)\rho(u;\varepsilon)\right]^{(n)},$$ which leads to the Jacobi polynomials $P_{n}^{(b_n-a_n,b_n+a_n)}\left(u\right)$. In this way, $\psi_n(z)$ in Eq. \eqref{eigenfunctionsSchroMF} are possible solutions of Problem \ref{schroProb} for the bound states region. 

Now, we will prove that $(\varepsilon_n,\psi_n(z))$ are the only non-trivial solutions of the eigenvalue Schrödinger problem for the bound states region. 
To begin with, we note that every $\psi_n(z) \in L^{2}(\mathbb{R})$. This can be easily verified since the Jacobi polynomials are bounded given that $\tanh(z) \in (-1,1)$ and $a_n$, $b_n$, $b_n \pm a_n > 0$ implies that $e^{-a_n z} \sech^{b_n}{z}$ is square-integrable in $\mathbb{R}$ for every $\varepsilon_n \in (0,v_-)$ which ultimately leads to $\psi_n(z) \in L^2(\mathbb{R})$. 
In addition, this computation explains why the limit state $\varepsilon=v_{-}$ cannot be treated as part of the bound states region, since, by putting $\varepsilon_m=v_{-}$ as a possible bound eigenvalue, its associated $\psi_m(z)$ given by equation \eqref{eigenfunctionsSchroMF} would not be of integrable square in $\mathbb{R}$ as $a_m = a(v_-) = b(v_-) = b_m$.

In any case, given another possible solution, $\psi_{\varepsilon}(z) = \phi(-\tanh{z};\varepsilon)y(-\tanh{z};\varepsilon) \in L^{2}(\mathbb{R})$, it can be verified, by straightforward calculations, that
\begin{equation*}
\int_{-1}^{+1}\left|\sqrt{\rho(u;\varepsilon)}y(u;\varepsilon)\right|^2du = \int_{-\infty}^{+\infty} \left|\psi_{\varepsilon}(z)\right|^2 \sech^2{z}dz \le \int_{-\infty}^{+\infty}\left|\psi_{\varepsilon}(z)\right|^2 dz,
\end{equation*}
meaning square-integrability of $\psi_{\varepsilon}(z)$ implies $\sqrt{\rho(u;\varepsilon)}y(u;\varepsilon) \in L^{2}(-1,1)$.
Therefore, since previous $y(u;\varepsilon)$ would also be a solution of Eq. \eqref{hdeMF}, from the uniqueness of Theorem \ref{eigenTheo}, it follows that the functions given by Eq. \eqref{eigenfunctionsSchroMF} are the only non-trivials solutions of Eq. \eqref{schrEqMFz} belonging to $L^{2}(\mathbb{R})$ and corresponding to bound states.
\end{proof}

Before proceeding with the unbound states regions, the following remark is in order:
\begin{remark}
\label{finitenessOfBoundSolutionsRemark}
In principle, there might be no solutions for region \ref{R1} coming from no $n \in \mathbb{N} \cup \{0\}$ satisfying eigenvalue Eq. \eqref{eigenvalueEqMF}.

In effect, by making use of Eqs. \eqref{sys-abk}, it follows that $a_nb_n= \frac12 v_0 \sinh{2\mu}$ and $b_n^2+k_n=v_0\cosh^2\mu$. For solving this system, we combine last equation with Eq. \eqref{eigenvalueEqMF} which leads to $b_n^2+(2n+1)b_n+n(n+1)-v_0\cosh^2{\mu}=0$.
Therefore, since $b_n>0$, we obtain
\begin{equation} 
b_n=\sqrt{v_0\cosh^2{\mu}+\frac{1}{4}}-\left(n+\frac{1}{2}\right), \quad a_n=\frac{v_0 \sinh{2 \mu} }{2 b_n}.
\label{anbn}
\end{equation}
Now, since $b_n-a_n=\sqrt{v_--\varepsilon_n}>0$ for all $\varepsilon \in (0,v_-)$, then $b_n^2 > a_n\,b_n$, hence from  
Eq. \eqref{anbn} we get the inequality  
\begin{equation}
n < \sqrt{v_0\cosh^2{\mu}+\frac{1}{4}}-\sqrt{\frac{1}{2}v_0\sinh{2\mu}}-\frac{1}{2}=N(\mu,v_0), \quad n \in \mathbb{N} \cup \{0\},  
\label{nCondBounded}
\end{equation}
which allows us to calculate the number of bound solutions, $n_b$, by taking the largest non-negative integer smaller 
than $N(\mu,v_0)$.

Moreover, for all $v_0>0$ and $\mu\geq0$, the condition \eqref{nCondBounded} implies that 
$N(\mu,v_0)>0$, if and only if $v_0> v_{c}=e^{2\mu}\tanh \mu$. That is, under the  condition $v_0>v_c$ there will be at least one bound state whereas, if  $v_0 \le v_c$ there is not any bound state, see Fig.\ \ref{fig:vc}.

\begin{figure}[ht!]\centering
\includegraphics[width=8cm]{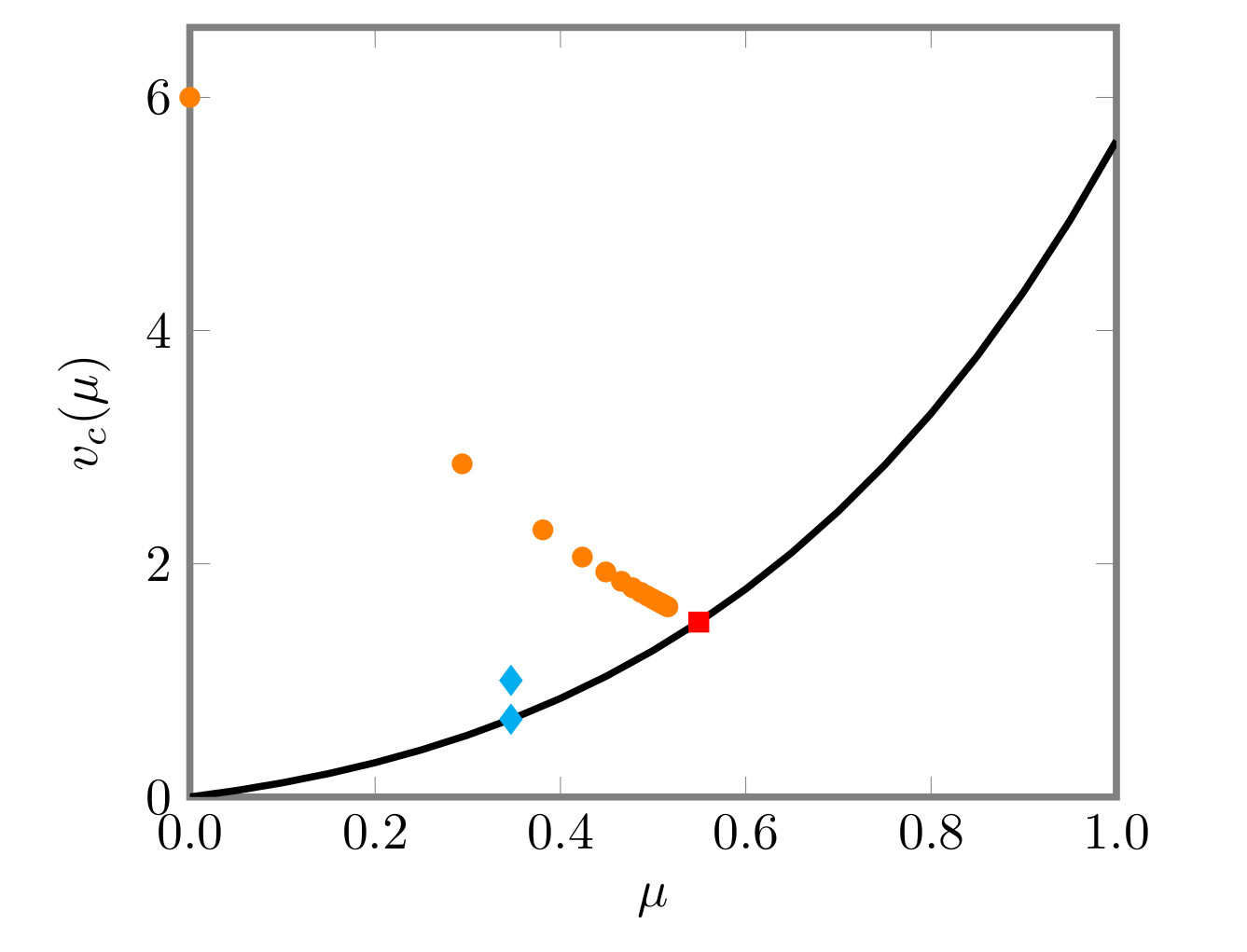}
	\caption{The curve represents the critical amplitude $v_c=e^{2\mu}\tanh \mu$ vs $\mu$. There will be at least one bound state for any point $(\mu,v_0)$ above this curve. Otherwise, there 
is no bound states. It is also shown  the points A$=(\ln(2)/2,1)$, 
B$=(\ln(2)/2,2/3)$ (blue diamonds), the values $(\mu,v_0)$ defined in Eq.\ \eqref{eq:para} (orange circles for $p=1,3,4,5,\dots $), and the red square represents its limit value, corresponding to the examples discussed in Section \ref{sec3}.
} \label{fig:vc}
\end{figure}

In addition, from above, it also follows that there is only a finite number of bound states of the Schrödinger Eq. \eqref{schrEqMFz} as stated on the previous theorem.

Lastly, by using the Eqs. \eqref{varkappa} and \eqref{R1ab}, it follows that 
\begin{equation}
\label{e_n}
\varepsilon_n=v_0e^{-2\mu}-(b_n-a_n)^2,\quad n=0,1,\dots,n_b,
\end{equation}
where $a_n$ and $b_n$ are given in Eq. \eqref{anbn}.

\end{remark}

\subsection{Unbound states regions \ref{R2} and \ref{R3}, $\varepsilon \in (v_-,v_+) \cup (v_+, \infty)$}
For the unbound states region, $\varepsilon \in (v_-,v_+) \cup (v_+,\infty)$, we are dropping the square-integrability condition on the eigenfunctions, $\psi$, and substituting it with a boundedness condition. Therefore, the rest of the solutions for Problem \ref{schroProb} must be found in a different way.  For the sake of simplicity we will write the solutions in this region by using the Gauss's hypergeometric function \eqref{hypergeometricSeries}, so that they can be compared with the ones obtained in Ref. \cite[\S12.3, pages 1651--1653]{morse2:1953}.

To that purpose, we first transform equation \eqref{hdeMF} into a Gauss's HDE \eqref{gaussEq} by setting $u=-1+2s$. Recall we were dealing with the HDE associated to the pair $(\tau_1(u;\varepsilon),\lambda_1(\varepsilon))$ for reasons related to applying Theorem \ref{eigenTheo} which, for these regions, is no longer useful. In any case, for the unbound regions, solving any of the four HDE is analogous as we won't be making use of sophisticated results. In this way, it becomes into
\begin{equation}
\begin{split} s(1-s)y''(s)  + \big[-2(b(\varepsilon)+1)s+  b(\varepsilon)+a(\varepsilon)+1\big]y'(s)
+ (k(\varepsilon)-b(\varepsilon))y(s)=0,
\end{split}
\label{GhdeMF}
\end{equation}
where $s \in (0,1)$, and, as before, $k(\varepsilon)$, $a(\varepsilon)$ and $b(\varepsilon)$ are defined in the region \ref{R2} by
\begin{equation}\label{R2ab}
k(\varepsilon)=\frac{\varepsilon+v_0}{2}-i\frac{\varkappa_+(\varepsilon) |\varkappa_-(\varepsilon)|}{2},\quad
a(\varepsilon)=\frac{\varkappa_+(\varepsilon) - i |\varkappa_-(\varepsilon)|}{2},\quad b(\varepsilon)=\frac{\varkappa_+(\varepsilon) + i |\varkappa_-(\varepsilon)|}{2},
\end{equation}
and in the region \ref{R3} by
\begin{equation}\label{R3ab}
k(\varepsilon)=\frac{\varepsilon+v_0}{2}+\frac{|\varkappa_+(\varepsilon) \varkappa_-(\varepsilon)|}{2},\quad
a(\varepsilon)=i \frac{|\varkappa_+(\varepsilon)| -  |\varkappa_-(\varepsilon)|}{2},\quad b(\varepsilon)=i \frac{|\varkappa_+(\varepsilon)| +  |\varkappa_-(\varepsilon)|}{2},
\end{equation}
while $\varkappa_+(\varepsilon)$ and $\varkappa_-(\varepsilon)$ are still defined by Eq.\ \eqref{varkappa}. 

Now, by comparing the Eq. \eqref{GhdeMF} with Eq. \eqref{gaussEq}, we can identify the parameters of the hypergeometric equation as follows
\begin{equation}
\begin{split}
\alpha(\varepsilon) = b(\varepsilon)+\frac{1}{2}-\sqrt{v_0\cosh^{2}{\mu}+\frac{1}{4}},& \quad \beta(\varepsilon) = b(\varepsilon)+\frac{1}{2}+\sqrt{v_0\cosh^{2}{\mu}+\frac{1}{4}}, \\ 
\gamma(\varepsilon) = a(\varepsilon)&+b(\varepsilon)+1.
\label{abgGausshdeMF}
\end{split}
\end{equation} 
In this way, our possible solutions take the following form in the $s$ variable
\begin{equation}
\psi(s) \propto \phi(-1+2s;\varepsilon)y(s;\varepsilon) =\left(1-s\right)^{\frac{b(\varepsilon)-a(\varepsilon)}{2}}s^{\frac{b(\varepsilon)+a(\varepsilon)}{2}}y(s;\varepsilon)
\label{waveFormGhdeMF}
\end{equation}
where $\phi(u;\varepsilon)$ was introduced in Eq. \eqref{rhophiMF}, and $y(s;\varepsilon)$ solves Eq. \eqref{GhdeMF}. 
 It is important to recall that, in general, there are two linearly independent solutions of Eq.\ \eqref{GhdeMF} and therefore
there will be two solutions of the form \eqref{waveFormGhdeMF}.

\begin{remark}
\label{disclaimerGamma}
In addition, we must point out that for the proofs of the following Theorem \ref{theoUnbound} and Theorem \ref{theoUnbound1}, we shall make use of certain asymptotic relations relating to the behaviour of Gauss's hypergeometric functions at $s=1^{+}$ extracted from \cite[\S15.4(ii)]{olver:2010}. However, the gamma function, $\Gamma(z)$, of different parameters ($\alpha, \, \gamma, \gamma-\alpha, \dots$) are involved in these relations and this may render the stated relation invalid as $\Gamma(z)$ is undefined in $\mathbb{Z}^{-}\cup\{0\}$. In any of the cases, it can be checked that, for every $\varepsilon$, all parameters involved in a employed relation in the proofs of both theorems have either a positive real part or a nonzero imaginary part except, perhaps, for $\varepsilon=v_-$ for which $\alpha(v_-)$ may (and can) be negative. In this way, unless stated otherwise, it must be assumed that there are no problems related to the gamma functions.
\end{remark}

With these ideas in mind, we present the theorem solving the unbound states regions \ref{R2} and \ref{R3}.

\begin{theorem}
\label{theoUnbound}
The non-trivial, unbound solutions (for a proper definition of the normalized unbound solutions leading to the so-called completeness relation, see Section \ref{app-close}) of time-independent Schrödinger equation \eqref{schrEqMFz}, $(\varepsilon, \psi_{\varepsilon}(z))$, for  $\varepsilon \in (v_-,v_+) \cup (v_+,\infty)$	satisfying that $\psi_{\varepsilon}(z) \in L^{\infty}(\mathbb{R})$, that is, the solutions of Problem \ref{schroProb} for the unbound states regions, read 
\begin{equation}
\psi_{\varepsilon}(z) = C \psi_1(z;\varepsilon), \quad C \in \mathbb{C},
\label{solRefStates}
\end{equation}
for region \ref{R2}, this is, for $\varepsilon \in (v_-,v_+)$, and 
\begin{equation}
\psi_{\varepsilon}(z) = C_1\psi_1(z;\varepsilon)+C_2\psi_2(z;\varepsilon), \quad C_1,\,C_2 \in\mathbb{C},
\label{solFreeStates}
\end{equation}
for region \ref{R3}, this is, for $\varepsilon \in (v_+,\infty)$, where $\psi_1(z;\varepsilon)$ and $\psi_2(z;\varepsilon)$ are given, respectively, by
\begin{equation}
\begin{split}
\label{solUnboundEq}
\psi_{1}(z;\varepsilon) = & \frac{\sech^{b(\varepsilon)}{z}}{e^{a(\varepsilon) z}}  F\left(\alpha(\varepsilon),\beta(\varepsilon),\gamma(\varepsilon);\frac{1-\tanh{z}}{2}\right), \\
\psi_{2}(z;\varepsilon) = &  \frac{e^{b(\varepsilon) z}}{\sech^{a(\varepsilon)}{z}}   F\left(\alpha(\varepsilon)-\gamma(\varepsilon)+1,\beta(\varepsilon)-\gamma(\varepsilon)+1,2-\gamma(\varepsilon);\frac{1-\tanh{z}}{2}\right),
\end{split}
\end{equation}
where $F(\alpha,\beta,\gamma;s)$ denotes the Gauss's hypergeometric function (see Proposition \ref{propGaussHypergeometricFunction} for more details),
$a(\varepsilon)$, $b(\varepsilon)$, $\alpha(\varepsilon)$, $\beta(\varepsilon)$, and $\gamma(\varepsilon)$, are the parameters defined for each region in Eqs. 
\eqref{R2ab}--\eqref{R3ab}, and \eqref{abgGausshdeMF}, respectively.
\end{theorem}

\begin{proof}
For region \ref{R3}, this is, whenever $\varepsilon \in (v_+,\infty)$, $\gamma(\varepsilon)=1+i|\varkappa_{+}(\varepsilon)| \notin \mathbb{R}$ which, accordingly to \cite[Table 13, \S21.4, page 281]{nu:1988} and given that
$s\in(0,1)$, implies that the functions 
\begin{equation} 
\begin{split}
\label{u1u2}
y_1(s;\varepsilon) & = F(\alpha(\varepsilon),\beta(\varepsilon),\gamma(\varepsilon);s)=w_1(s;\varepsilon), \\
y_2(s;\varepsilon) & =  s^{1-\gamma(\varepsilon)}F(1+\alpha(\varepsilon)-\gamma(\varepsilon),1+\beta(\varepsilon)-\gamma(\varepsilon),2-\gamma(\varepsilon);s)=s^{1-\gamma(\varepsilon)}w_2(s;\varepsilon),
\end{split}
\end{equation}
are two linearly independent solutions of the Eq. \eqref{GhdeMF}. Also, both $w_i(s;\varepsilon)$ are bounded whenever $\varepsilon \in (v_-,\infty)$, which includes region \ref{R3}, given that $$\gamma(\varepsilon)-\alpha(\varepsilon)-\beta(\varepsilon)=a(\varepsilon)-b(\varepsilon) = -i\left|\varkappa_-(\varepsilon)\right| \neq 0$$ in addition to relation 15.4.22 in  \cite[\S15.4(ii)]{olver:2010}. Thus, the functions 
\begin{equation}
\begin{split}
\psi_1(z;\varepsilon) & = \left(1-s(z)\right)^{i\frac{\left|\varkappa_-(\varepsilon)\right|}{2}}s(z)^{i\frac{\left|\varkappa_+(\varepsilon)\right|}{2}} w_1(s(z);\varepsilon), \\
\psi_2(z;\varepsilon) & = \left(1-s(z)\right)^{i\frac{\left|\varkappa_-(\varepsilon)\right|}{2}}s(z)^{-i\frac{\left|\varkappa_+(\varepsilon)\right|}{2}}w_2(s(z);\varepsilon),
\label{posWaveFunctionMFs}
\end{split}
\end{equation}
where $s(z)=(1-\tanh z)/2$, are two linearly independent and bounded solutions of the Eq. \eqref{schrEqMFz} for region \ref{R3}. This justifies the solution \eqref{solFreeStates}.

Now, for the region \ref{R2}, this is, whenever $\varepsilon \in (v_-,v_+)$, $\gamma(\varepsilon)=\sqrt{v_+-\varepsilon}+1 \in \mathbb{R}^{+}$ and it could be a positive integer greater than one. Therefore, the second solution, $y_2(s;\varepsilon)$, in \eqref{u1u2} could be undefined as $2-\gamma(\varepsilon)$ could be a negative integer (see Proposition \ref{propGaussHypergeometricFunction}). Actually, this is the case whenever $ \varepsilon = v_{+}-m^2 \eqqcolon \varepsilon_m$ and $m = 1,2,\ldots,q$, being $q$ the greatest positive integer satisfying the condition $1 \le q <\sqrt{v_+-v_-}$. For the cases $\varepsilon=\varepsilon_m$, we will need to find another linearly independent solution, $y_2(s;\varepsilon)$.

Let us first deal with $\varepsilon\in(v_-,v_+)$ and $\varepsilon\neq \varepsilon_m$, that is,
whenever $\gamma(\varepsilon)$ is not a positive integer. In this case, the functions given in Eq. \eqref{posWaveFunctionMFs} with its parameters particularized to region \ref{R2} are also linearly independent solutions of the Eq. \eqref{schrEqMFz}. In particular, for the region \ref{R2}, they are written in the form
\begin{equation}
\begin{split}
\psi_1(z;\varepsilon) & = \left(1-s(z)\right)^{i\frac{\left|\varkappa_-(\varepsilon)\right|}{2}}s(z)^{\frac{\varkappa_+(\varepsilon)}{2}}w_1(s(z);\varepsilon), \\
\psi_2(z;\varepsilon) & = \left(1-s(z)\right)^{i\frac{\left|\varkappa_-(\varepsilon)\right|}{2}}s(z)^{-\frac{\varkappa_+(\varepsilon)}{2}}w_2(s(z);\varepsilon),
\label{posWaveFunctionR2}
\end{split}
\end{equation}
meaning $\psi_1(s;\varepsilon)$ is bounded and $\psi_2(s;\varepsilon)$ is unbounded at $s = 0^{+}$ since, recall, both $w_i(s;\varepsilon)$ were bounded for every $\varepsilon \in (v_-,\infty)$. Therefore, for the region \ref{R2}, except, perhaps, for $\varepsilon=\varepsilon_m$, 
there is only one unbound solution of Eq. \eqref{schrEqMFz} which is, precisely, $\psi_{1}(s;\varepsilon)$. This partially justifies \eqref{solRefStates}. Let us show that this is still the case for $\varepsilon=\varepsilon_m$.

In this way, we begin by recalling that the $\varepsilon_m \in (v_-,v_+)$ for $m=1,2,\ldots,q$ satisfy that  $\gamma(\varepsilon_m)=m+1$
is a positive integer. In this case, we insist that $\psi_1(z;\varepsilon)$ given in Eq. \eqref{posWaveFunctionR2} is still a possible solution of 
Eq. \eqref{schrEqMFz} (it always is) which is bounded for the very same reasons as before. Therefore, we only need to find a second linearly independent solution and show its unbounded character. 

In this regard, accordingly to \cite[\S21.4, page 277]{nu:1988}, since, $\gamma(\varepsilon_m)$ is a positive integer,
there are two possibilities for a fixed $m$:
\begin{itemize}
\item[A.] $\alpha(\varepsilon_m)$,  $\beta(\varepsilon_m)$, and $\alpha(\varepsilon_m)+\beta(\varepsilon_m)$ are not integers, or
\item[B.] one of them is an integer.
\end{itemize}
From Eqs. \eqref{R2ab} and taking into account that  $\varepsilon_m\not =v_-$, 
it follows that $\Im(\alpha(\varepsilon_m))\neq0$, $\Im(\beta(\varepsilon_m))\neq0$, and $\Im(\alpha(\varepsilon_m)+\beta(\varepsilon_m))\neq0$ for every $m=1,2,\dots,q$, meaning we deal with the case A. Therefore, the second solution of Eq. \eqref{GhdeMF} is given by $y_A(s;\varepsilon_m)=F(\alpha(\varepsilon_m),\beta(\varepsilon_m),\alpha(\varepsilon_m)+\beta(\varepsilon_m)-m;1-s)$ (see \cite[\S21.4, page 277]{nu:1988}). 
Thus, a second linearly independent solution of Eq. \eqref{schrEqMFz}
is given by
$$
\psi_A(s;\varepsilon_m)=\left(1-s\right)^{i\frac{\left|\varkappa_-(\varepsilon_m)\right|}{2}}s^{\frac{\varkappa_+(\varepsilon_m)}{2}} y_A(s;\varepsilon_m).
$$
However, accordingly to relation 15.4.23 in  \cite[\S15.4(ii)]{olver:2010},
$$
F(\alpha,\beta,\alpha+\beta-m;1-s) \sim \frac{\Gamma(\alpha+\beta-m)\Gamma(m)}{\Gamma(\alpha)\Gamma(\beta)} s^{-m} \text{ as } s \to 0^{+},
$$
from where it follows that $\psi_A(s,\varepsilon_m)$ is unbounded at $s=0^{+}$ since $\varkappa_+(\varepsilon_m)=\sqrt{v_+-\varepsilon_m}=m$. 
Thus, we have only one bounded solution. This completes the proof.
\end{proof}

\begin{remark} \label{complexConjugateRemark}
Notice that, for the region \ref{R3}, by making use of the values in Eq.\ \eqref{R3ab} as well as the identity $F(\alpha,\beta, \gamma ;z)=(1-z)^{\gamma-\alpha-\beta}\,F(\gamma-\alpha,\gamma-\beta,\gamma;z)$ (see \cite[Eq. 15.8.1 \S15.8(i)]{olver:2010}) 
it can be shown that the unbound solutions in Eq.\ \eqref{solUnboundEq}  satisfy the relation $\psi_1(z;\varepsilon)=\psi_2^\star(z;\varepsilon)$, where 
${z^\star}$ denotes the complex conjugate of $z$. 
\end{remark}

Now, before proceeding with the particularization to the symmetric case, we shall deal with the special states $\varepsilon = 0, v_{\pm}$.

 \subsection{The special states, $\varepsilon = 0,\, v_{\pm}$}
 
Let us now investigate the cases $\varepsilon=0$, and $\varepsilon=\,v_{\pm}$. 
To begin with, let us point out that if $\varepsilon=0$ were a possible energy state, then it would necessarily need to be associated to a bound solution, meaning that it should be a possible solution of the eigenvalue Eq. \eqref{eigenvalueEqMF}. However, setting $\varepsilon=0$ in Eq. \eqref{R1ab} yields $k(0)=0$ and $b(0)>0$ which ultimately implies that Eq. \eqref{eigenvalueEqMF} has no solutions for that case. Therefore, in the end, the energy for the system described by Eq. \eqref{schrEqMFz} is always positive as is to be expected.

Regarding, cases $\varepsilon = v_{\pm}$, we shall treat them both as part of the unbound states regions. We already justified this choice for $\varepsilon=v_-$ in the proof of Theorem \ref{theoBounded} and it is somewhat obvious for $\varepsilon=v_+$. In this fashion, we obtain the following result that completes previous Theorem \ref{theoUnbound}.

\begin{theorem} \label{theoUnbound1}
Concerning the special states $\varepsilon = v_{\pm}$ and the continuous spectrum, $\varSigma$, of Problem \ref{schroProb}, the following statements are satisfied:
\begin{enumerate}[label=\text{(\arabic*)}]
\item $v_+ \in \varSigma$ and its unbound solution is written in the form
\begin{equation} \label{psivplus}
\psi_{v_+}(z) \propto \psi_1(z;v_+),
\end{equation}
where $\psi_1(z;\varepsilon)$ is defined in Eq.  \eqref{solUnboundEq}.
\item $v_- \in \varSigma$ if and only if $\alpha(v_-) = -l$ or, equivalently, $N(\mu,v_0)=l$ for some $l = 0,1,2,\dots$, where $N(\mu,v_0)$ is given by Eq.\ \eqref{nCondBounded}. In that case, its unbound solution, $\psi_{v_-}(z)$, reduces to
\begin{equation}
\label{particularCasev-state}
    \psi_{v_-}(z) \propto e^{-a(v_-)z} \sech^{a(v_-)z}(z) P_{l}^{(0,2a(v_-))}\left(-\tanh{z}\right),
\end{equation}
being $a(v_-) = \frac12\sqrt{v_+-v_-}$.
\end{enumerate}
\end{theorem}
\begin{proof}
We reason in a similar fashion as in previous theorem. We start by recalling that $\psi_1(z;\varepsilon)$ given in Eq. \eqref{solUnboundEq} is a possible solution to Eq. 
\eqref{schrEqMFz} for both cases $\varepsilon = v_{\pm}$. All that is left is to find a second linearly independent solution and justify the boundedness or unboundedness of both solutions.

We begin dealing with the case $\varepsilon=v_+$. In this case, $\varkappa_{+}(v_+)=0$, $\varkappa_{-}(v_+)=i\sqrt{v_+-v_-}\neq0$, and $\gamma(v_+)=1$. Therefore, the first solution 
of Eq. \eqref{schrEqMFz} is given by $$\psi_1(s;v_+)=\left(1-s\right)^{i\frac{|\varkappa_-(v_+)|}{2}}F(\alpha(v_+),\beta(v_+),1;s),$$ which is bounded thanks to relation 15.4.22 in  \cite[\S15.4(ii)]{olver:2010}. 
In this case, since $\alpha(v_+)$,  $\beta(v_+)$, and $\alpha(v_+)+\beta(v_+)$ are not integers and $\gamma(v_-)$ is an integer, we are dealing with the case A described in the proof of Theorem \ref{theoUnbound}
and accordingly to \cite[page 277]{nu:1988},  the second solution is given by
$$
\psi_{A}(s;v_+)=\left(1-s\right)^{i\frac{|\varkappa_-(v_+)|}{2}}F(\alpha(v_+),\beta(v_+),\alpha(v_+)+\beta(v_+);1-s).
$$
Accordingly to relation 15.4.21 in  \cite[\S15.4(ii)]{olver:2010},
$$
F(\alpha,\beta,\alpha+\beta;1-s) \sim -\frac{\Gamma(\alpha+\beta)}{\Gamma(\alpha)\Gamma(\beta)} \ln(s) \text{ as } s \to 0^{+},
$$
meaning it is unbounded at $s=0^+$ and, therefore, there is only one eigenfunction corresponding to the eigenenergy $\varepsilon=v_+$ which is $\psi_1(s;v_+)$. 
In this way, we include this case in the region \ref{R2}. 

Let us finally consider $\varepsilon=v_-$. As we expressed earlier, $\psi_1(z;\varepsilon)$ is always a possible unbound solution for every $\varepsilon \in [v_-,\infty)$ and, in this case, it reads $$\psi_1(z;v_-) = s(z)^{\frac{\varkappa_+(v_-)}{2}}F(\alpha(v_-),\beta(v_-),\gamma(v_-);s(z)),$$ where $s(z)=(1-\tanh z)/2$. Now, using the values in Eq.\ \eqref{abgGausshdeMF}, it can be checked that $\alpha(v_-)+\beta(v_-)=\gamma(v_-)$ which coupled to relation 15.4.21 in  \cite[\S15.4(ii)]{olver:2010}, means that
\begin{equation}
\label{u1v-}
y_{1}(s;v_-)= F(\alpha(v_-),\beta(v_-),\gamma(v_-);s) \sim 
-\frac{\Gamma(\gamma(v_-))}{\Gamma(\alpha(v_-))\Gamma(\beta(v_-))}\ln(1-s), \text{ as } s \to 1^{-}.
\end{equation}
However, for the special case $\varepsilon=v_-$, the previous relation could be invalid (see Remark \ref{disclaimerGamma}) as $\alpha(v_-)$ could be a non-positive integer, say $\alpha(v_-)=-l$ for some $l=0,1,2,\dots$, which would result in an undefined $\Gamma(\alpha(v_-))$ and rendering Eq. \eqref{u1v-} incorrect.
Therefore, the solution $\psi_1(z;v_-)$ given in Eq. \eqref{solUnboundEq} will be bounded if and only if the series $F(\alpha(v_-),\beta(v_-),\gamma(v_-);s)$ 
is terminating or, equivalently, $\alpha(v_-)=-l$, for some $l=0,1,2,\dots$, otherwise it will unbounded  by \eqref{u1v-}. Also, it can be shown that there exists pairs of values $(v_0,\mu)$ for which $\alpha(v_-)=-l$ for some $l=0,1,2,\dots$ See Remark \ref{existenceOfv-state} for further explanations. 

Now, we shall find a second linearly independent solution and prove that it always is an unbounded function. Therefore, according to \cite[Table 13, \S21.4, page 281]{nu:1988} and, analogously to what has been before, there are two possibilities regarding $\gamma(v_-)$: \begin{itemize}
\item[C.] $\gamma(v_-) = 1 + \varkappa_+(v_-) =1 + m$ for some $m=1,2,\dots$, or
\item[D.] neither $\gamma(v_-)$ nor $\varkappa_{+}(v_-)$ are an integer.
\end{itemize}

Since, in principle, we know nothing of $v_+-v_-$, we need to exhaust both options. We start dealing with the case D in which $\varkappa_+(v_-) \notin \mathbb{Z}$. In this case, $\gamma(v_-)$ is not a positive integer, the two independent solutions of Eq. \eqref{schrEqMFz} are given again by  Eq. 
\eqref{posWaveFunctionR2} and it can be easily checked that $\psi_2(s;v_-)$ is unbounded at $s \to 0^+$ due to the factor $s^{-\frac12\varkappa_+(v_-)}$.

Next, we deal with the case C in which $\varkappa_+(v_-)=\sqrt{v_+-v_-}= m = 1,2,\dots $ This implies that $\gamma(v_-)= 1+m$ and, as it is shown in \cite[Table 13, \S21.4, page 281]{nu:1988}, there are another two possibilities regarding the two new parameters $\alpha'=\alpha-\gamma+1$ and $\beta'=\beta-\gamma+1$. These possibilities are: \begin{itemize}
\item[C.1] neither $\alpha'(v_-)$ nor $\beta'(v_-)$ is a equal to some $k=0,1,\dots,m-1$, or
\item[C.2] $\alpha'(v_-)=k$ or $\beta'(v_-)=k$ for some $k=0,1,\dots,m-1$,
\end{itemize}
where, recall that, we are assuming $\varkappa_+(v_-)=m$.
For case C.2, two independent solutions of Eq.\ \eqref{schrEqMFz} are given, once again, by  Eq.\ \eqref{posWaveFunctionR2} meaning $\psi_1(s;v_-)$ is the only possibly bounded solution.

Lastly, we deal with the case C.1. According to \cite[Table 13, \S21.4, page 281]{nu:1988}, for this specific situation, apart from $y_1(s;\varepsilon)$ given in \eqref{u1u2}, another linearly independent solution of \eqref{GhdeMF}, $y_2(s;\varepsilon)$ is the function $\Phi(\alpha(\varepsilon),\beta(\varepsilon),\gamma(\varepsilon),s)$ defined in \cite[Eq. (32), page 279]{nu:1988}. By making use of the explicit expression of $\Phi$, it can be checked that the function $$\psi_{C.1}(s;v_-)=s^{\frac12\varkappa_+(v_-)}\Phi(\alpha(v_-),\beta(v_-),\varkappa_+(v_-)+1,s),$$ which is linearly independent to $\psi_1(s;v_-)$, is unbounded at $s \to 0^{+}$, as one of its terms satisfies that, $m=\varkappa_{+}(v_-)$,
$$
\lim_{s \to 0^+}\sum^{m}_{k=1}\dfrac{\left(-1\right)^{k-1}\left(k-1\right)!}{\left(n-k\right)_{k}\left(\alpha-k\right)_{k}\left(\beta-k\right)_{k}}s^{\frac{m}{2}-k}=\infty,
$$
and the rest are bounded. In such fashion, we have exhausted all options and found that every linearly independent solution to $\psi_{1}(s;\varepsilon)$ is unbounded.

To conclude the proof, we explicitly compute solution $\psi_1(s;v_-)$ whenever $\alpha(v_-)=-l$ for some $l=0,1,2,\dots$  
From \eqref{posWaveFunctionR2},
and by using \eqref{jac-hyp}, we find that
$$
\psi_1(s,v_-)= s^{\frac{\varkappa_+(v_-)}2} F(-l,\beta(v_-),\alpha(v_-)+\beta(v_-);s)= s^{\frac{\varkappa_+(v_-)}2} P_{l}^{(\varkappa_+(v_-),0)}\left(1-2s\right),
$$
where $P_l^{(\varkappa_+(v_-),0)}$ denotes the Jacobi polynomials of degree $l$. Thus, by making use of the symmetric property of Jacobi polynomials (see first formula in \cite[page 41]{nu:1988}), the above expression becomes into \eqref{particularCasev-state}. It can be checked that this solution is never of integrable square, although its form coincides with the one predicted for bound solutions given by Eq. \eqref{eigenfunctionsSchroMF} in Theorem \ref{theoBounded}. 

Also, straightforward calculations show that $\alpha(v_-)=-N(\mu,v_0)$, meaning condition \eqref{nCondBounded} for existence on bound solutions could be relaxed by substituting ``$<$'' for ``$\le$'' and including the existence of $\varepsilon=v_-$ as an eigenenergy in the equality $l=N(\mu,v_0)$ for some $l \in \mathbb{N} \cup \{0\}$. Regardless these facts, we still consider $v_-$ as part of the unbound states region \ref{R2}. Thus, the theorem is proved.
\end{proof}

Before proceeding for the particularization for the symmetric case $\mu=0$, the following remarks are in order.

\begin{remark}
\label{existenceOfv-state}
We show that there always exists pairs of values $(v_0,\mu)$ for which $\alpha(v_-)=-l$ for some $l=0,1,2,\dots$ In order to do so, 
we rewrite the equation $\alpha(v_-)=-N(\mu,v_0)=-l$ as $g(\mu,v_0)-(2l+1)=0$, where
$$
g(\mu,v_0)= \sqrt{4v_0\cosh^2{\mu}+1}-\sqrt{{2}v_0\sinh{2\mu}}.
$$
By straightforward calculations, it follows that, for every $v_0>0$, $g(v_0,\mu)$ is a decreasing
function of $\mu$ and satisfies that $\lim_{\mu\to0+}g(\mu,v_0)=\sqrt{4v_0+1}$ and 
$\lim_{\mu\to+\infty}g(\mu,v_0)=0$. Therefore, for every $v_0>0$, there is always a value of $\mu$ for which $g(\mu,v_0)=1$. In fact, let $m$ be the greatest non-negative integer such that $v_0>m(m+1)$, then, there exist $\mu_l$ for which $g(\mu_l,v_0)=2l+1$, for every $l=0,1,\dots,m$. Therefore, for these pairs $(\mu_l,v_0)$ with $l=0,1,2, \dots,m$, $\varepsilon=v_{-}$ 
belongs to the continuous spectrum of Problem \ref{schroProb}. 
\end{remark}

\begin{remark}
From theorems \ref{theoUnbound} and \ref{theoUnbound1}, it follows that the continuous spectrum of the equation \eqref{schrEqMFz}, $\varSigma \subset [v_-, \infty)$, always satisfies that 
$(v_-,\infty) \subset \varSigma$ and reaches the equality $\varSigma = [v_-, \infty)$ whenever $\alpha(v_-)=-l=-N(\mu,v_0)$ for some $l=0,1,2,\dots$
\end{remark}

\subsection{The symmetric case, $\mu = 0$}
Now, we shall particularize above theorems for the symmetric case $\mu = 0$. In fact, by putting $\mu = 0$ in Eq. \eqref{schrEqMFz},  
we obtain, for the potential $v(z)$, the expression
\begin{equation}
v(z) = v_0 \tanh^2{z} = v_0 \left(1-\frac{1}{\cosh^2{z}}\right),  \label{pos-tel}
\end{equation}
which is clearly symmetric. The above potential is usually called the  P\"oschl-Teller potential and was introduced in \cite{eckart:1930,epstein:1930} 
(see also page 768 in Ref. \cite{morse:1953}).
The symmetric behaviour implies that $v_+ = v_-  = v_0$, meaning the reflecting states region is lost. Therefore, in this case, we may only distinguish two ranges of 
eigenenergies leading to the two following definitions of states regions.
\begin{enumerate}[label=\text{S\arabic*}]
\item \label{sR1} For $\varepsilon \in (0,v_0)$, we distinguish the bound states region.
\item \label{sR2} For $\varepsilon \in (v_0,\infty)$, we distinguish the free states region. 
\end{enumerate}
In this way, we obtain the following corollaries for $\mu=0$ as result of Theorems \ref{theoBounded} and \ref{theoUnbound}.

\begin{corollary}
The only non-trivial, bound solutions of time-independent Schrödinger equation with the P\"oschl-Teller potential \eqref{pos-tel}, $(\varepsilon,\psi_{\varepsilon}(z))$, for $\varepsilon \in (0,v_0)$ satisfying that 
$\psi_{\varepsilon}(z) \in L^2(\mathbb{R})$, that is, the solutions of the symmetric Problem \eqref{schroProb} 
for region \ref{sR1}, are
\begin{equation*}
\psi_{n}(z) = \mathcal{N}_n \sech^{b_n}{z} P_{n}^{(b_n, \; b_n)}\left(-\tanh{z}\right) = \psi_n(z),
\end{equation*}
where $\mathcal{N}_n$ is a normalizing constant such that $\int_{\mathbb{R}}|\psi_n(z)|^2dz=1$ and the eigenvalue $\varepsilon = \varepsilon_n \in (0,v_0)$ satisfy the relation
\begin{equation*}
\begin{gathered}
\varepsilon_n = v_0 - b_n^2, \quad b_n = \sqrt{v_0+\frac{1}{4}}-n-\frac{1}{2}, \quad n \in \mathbb{N} \cup \{0\},
\end{gathered}
\end{equation*}
valid whenever $b_n>0$.
\end{corollary}

\begin{corollary}
The non-trivial, unbound solutions of time-independent Schrödinger equation  with the P\"oschl-Teller potential \eqref{pos-tel}, $(\varepsilon, \psi_{\varepsilon}(z))$, for $\varepsilon \in [v_0,\infty)$ satisfying that $\psi_{\varepsilon}(z) \in L^{\infty}(\mathbb{R})$, that is, solves the symmetric Problem \ref{schroProb}, are given by
\begin{equation*}
\psi_{\varepsilon}(z) = C_1\psi_1(z;\varepsilon)+C_2\psi_2(z;\varepsilon), \quad C_1, \, C_2 \in\mathbb{C},
\label{solFreeStatesSym}
\end{equation*}
where
\begin{equation*}
\begin{split}
	\psi_{1}(z;\varepsilon) &= \sech^{b(\varepsilon)}z 
	F\left(-b_0+i\sqrt{\varepsilon-v_0},b_0+1+i\sqrt{\varepsilon-v_0},1+i\sqrt{\varepsilon-v_0};\frac{1-\tanh{z}}{2}\right)
	, \\
	\psi_{2}(z;\varepsilon) &= e^{b(\varepsilon) z}
	F\left(-b_0,1+b_0,1-i\sqrt{\varepsilon-v_0});\frac{1-\tanh{z}}{2}\right),
\end{split}
\end{equation*} and $b(\varepsilon) = i\sqrt{\varepsilon-v_0}$. 
Moreover, the continuous spectrum, $\varSigma \subset [v_0, \infty)$, always fills the open interval, this is, $(v_0,\infty) \subset \varSigma$, 
and only reaches $\varepsilon = v_0$ whenever $\alpha(v_0)=-b_0=-l$, that is, whenever $v_0=l(l+1)$ for some $l=1,2,\ldots$.
\end{corollary}
In fact, whenever $v_0=l(l+1)$, $l\in\mathbb{N} \setminus \{0\}$ we recover the recent results of Ref. \cite{raban:2022}. Note that by setting $l=0$, which is valid for the asymmetric case, $v_0$ would be null and would result in no potential at all. For these reasons, $\alpha(v_0)=0$ is not a possibility. Notice also that, whenever $\varepsilon = v_0$, $\psi_1(z;\varepsilon)$ and $\psi_2(z;\varepsilon)$ are linearly dependent.

We point out that, as in the case $\mu>0$, the relation $\psi_1(z;\epsilon)=\psi_2^{\star}(z;\epsilon)$, where
${z^\star}$ denotes the complex conjugate of $z$, still holds.

\subsection{The orthogonality and completeness relations \label{sec4}}
\label{app-close}
Next, we will discuss the orthogonality and completeness relations for the solutions of the Eq. \eqref{schrEqMFz} with $\mu \neq 0$. The particular case of $\mu=0$ has been recently discussed in Ref. \cite{raban:2022}.   
For doing that, we will follow the theory of the one-dimensional Schr\"odinger equation, 
developed in Chapter 3\S2 of Ref. \cite{takh:2008}. More exactly we will use the results discussed in Example 2.1 page 177 in 
\cite{takh:2008} adapted to our case. 

First of all, for $\varepsilon\geq v_-$, we need to define the so-called Jost functions which are two solutions of  Eq. \eqref{schrEqMFz}
that, in our case, should have the asymptotics
\begin{eqnarray}  \label{f1}
	f_1(z;k) & = &  e^{-ikz}+o(1)\mbox{ as } z\to-\infty, \quad    k=\sqrt{\varepsilon-v_-},\\
	\label{f2}
	f_2(z;k) & =  & e^{ik_1z}+o(1)\mbox{ as } z\to+\infty, \qquad   k_1=\sqrt{\varepsilon-v_+}.
\end{eqnarray}
From Theorem \ref{theoUnbound} and using equation 15.10.21 in  \cite[\S15.10(ii)]{olver:2010}, 
following the ideas in section \textit{Asymptotic behaviour for the unbound solutions} of Ref. \cite[pag 62-64]{gor:2022}, it can be shown that, for $\varkappa=\sqrt{v_+-v_-}$,
\begin{equation}\label{f1y2}
	\begin{split} 
		f_1(z;k) &= \frac{\Gamma'_1(\varepsilon)\left(2^{-b(\varepsilon_k)}\Gamma_2(\varepsilon_k)\psi_1(z;\varepsilon_k)-2^{a(\varepsilon_k)}\Gamma_1(\varepsilon_k)\psi_2(z;\varepsilon_k)\right)}{\Gamma_2(\varepsilon)\left[\Gamma_2(\varepsilon)\Gamma'_1(\varepsilon)-\Gamma_1(\varepsilon)\Gamma'_2(\varepsilon)\right]}, \\
		f_2(z;k) &=  2^{-b(\varepsilon_k)}\mathds{1}_{(0,\varkappa)}(k)\psi_1(z,\varepsilon_k)+2^{a(\varepsilon_k)}\mathds{1}_{(\varkappa,\infty)}(k) \psi_2(z,\varepsilon_k),
	\end{split}   
\end{equation}
where $ \varepsilon_k = k^2+v_-$, $\Gamma_{1,2}(\varepsilon) := \Gamma_{1,2}(\alpha(\varepsilon),\beta(\varepsilon),\gamma(\varepsilon))$ and $\Gamma'_{1,2}(\varepsilon) := \Gamma_{1,2}(\alpha'(\varepsilon),\beta'(\varepsilon),\gamma'(\varepsilon))$, 
being 
\begin{equation}
	\label{Gamma12} 
	\Gamma_1(\alpha,\beta,\gamma) = \frac{\Gamma(\gamma)\Gamma(\gamma-\alpha-\beta)}{\Gamma(\gamma-\alpha)\Gamma(\gamma-\beta)},   \qquad \Gamma_2(\alpha,\beta,\gamma) = \frac{\Gamma(\gamma)\Gamma(\alpha+\beta-\gamma)}{\Gamma(\alpha)\Gamma(\beta)},
\end{equation}
where $\alpha' = \alpha-\gamma+1$, $\beta' = \beta-\gamma+1$, $\gamma' = 2-\gamma$,
and, finally,
\begin{equation*}
	\mathds{1}_{J}(k)=
	\begin{dcases}
		1, & k \in J, \\
		0, & k \notin J.
	\end{dcases}
\end{equation*}
Lastly, we compute the so-called \textit{transmission} coefficient $c(k)$, 
$$
c(k):=\frac{1}{2ki} W[f_1(z;k),f_2(z;k)]=\frac{X(\varepsilon_k)\left(\mathds{1}_{(0,\varkappa)}(k)\Gamma_1(\varepsilon_k)+\mathds{1}_{(\varkappa,\infty)}(k)\Gamma_2(\varepsilon_k)\right)}{4ik}\left(a(\varepsilon_k)+b(\varepsilon_k)\right),
$$
where the formula (15.10.3) in  \cite[\S15.10(i)]{olver:2010} has been used.  

In this way, from the results of \cite[Example 2.1 page 177]{takh:2008} it follows the next theorem that also summarize what has been previously obtained.

\begin{theorem} \label{theo2.4}
	The spectrum of the Sturm-Liouville equation \eqref{schrEqMFz} is described as follows:
	
	\begin{itemize}
		\item There are a finite number, $n_b$, of eigenvalues, $\varepsilon_n\in(0,v_-)$, corresponding to the discrete spectrum of the Sturm-Liouville problem \eqref{schrEqMFz} given by the equation \eqref{eigenvalueEqMF}
		with normalized eigenfunctions, $\psi_n$, defined in Eq. \eqref{eigenfunctionsSchroMF}.
		\item The absolutely continuous spectrum of the Sturm-Liouville problem \eqref{schrEqMFz} fills $(v_-,+\infty)$. 
		For $k \in (0,\varkappa)$, $\varkappa= \sqrt{v_+-v_-}$, that is, for $\varepsilon\in (v_-,v_+)$, the spectrum is simple and the 
		normalized eigenfunction is given by
		$$
		u_0(z;k)=\frac{f_2(z;k)}{c(k)}.
		$$For $k\in(\varkappa,\infty)$, that is, for $\varepsilon\in(v_+,\infty)$, the spectrum has multiplicity two and the normalized 
		eigenfunctions are
		$$
		u_1(z;k)=\frac{k_1}{k \, c(k)}f_1(z;k),\qquad u_2(z;k)=\frac{f_2(z;k)}{c(k)}, 
		$$
		where $k$ and $k_1$ are given in Eqs. \eqref{f1} and \eqref{f2}, respectively, 
		and $f_{1,2}$ are given by \eqref{f1y2}.
	\end{itemize}
	
	Moreover, for all
	$\Phi(x)\in L^2(\mathbb{R})$ we have the expansion (recall $\varkappa= \sqrt{v_+-v_-}$)
	
	\begin{eqnarray}\label{exp-u-psi}
		\Phi(z)= & 
		\displaystyle  \int_{0}^{\varkappa} \!\!  c_{0}(k) u_{0}(z;k)dk   +
		\sum_{\ell=1}^2 \int_{\varkappa}^{\infty} \!\!  c_{\ell}(k) u_{\ell}(z;k)dk   
		+ \displaystyle \sum_{\kappa=1}^{n_b}  c_\kappa  \psi_\kappa(z),
	\end{eqnarray}
	where
	\begin{eqnarray} \nonumber
		c_\ell(k)= \int_{\mathbb{R}} {u_\ell^\star(y,k)}\Phi(y)dy, \quad \ell=0,1,2,\quad && 
		c_\kappa= \int_{\mathbb{R}} \psi_\kappa(y)\Phi(y)dy. 
	\end{eqnarray}
	Formula \eqref{exp-u-psi} is the so-called completeness relation for the set $\{u_0,u_1,u_2\}_{k\geq0}
	\cup \{\psi_{\kappa}\}_{\kappa=1,\ldots,n_b}$.
	
	Finally, the eigenfunctions constitute an orthogonal system, that is,
	\begin{equation}
		\label{rel-ort-psi-b}\int_{\mathbb{R}} \psi_\kappa(z) \psi_\nu(z)dz=\delta_{\kappa,\nu},\quad 
		\int_{\mathbb{R}} \psi_\kappa(z)  {u_\ell(z;k)}dx=0,   \quad
	\end{equation}
	\begin{equation} 
		\label{rel-ort-psi-no-con} \frac{1}{2\pi}
		\int_{\mathbb{R}} {u_\ell^\star(z;k)} u_\iota(z,m)dx=\delta_{\ell,\iota}\delta (k-m),  
	\end{equation}
	where $\nu,\kappa=1,2\ldots,n_b$; $\ell,\iota=0,1,2$; $k,m\geq0$; $\delta_{\kappa,\nu}$ is the Kronecker delta, 
	and  $\delta(x)$ denotes the delta Dirac function (which is not a function, but a distribution,
	so Eq. \eqref{rel-ort-psi-no-con} should be understood in the distributional sense \cite{strichartz:1994}).
	As before  ${z^\star}$ denotes the complex conjugate of $z$.
\end{theorem}

\section{Applications \label{sec3}}

In order to apply all of these results, we discuss some representative examples. In the first of them, we choose $(\mu,v_0)$ such that the spectrum is represented by discrete and continuous parts. For the second choice of $(\mu,v_0)$, 
the state $\varepsilon=v_{-}$ belongs to the continuous spectrum and there are no bound states. The third example culminates this section by analyzing the linear stability of a family of kinks.   

\subsection{Discrete and continuous spectrum}

By setting $\mu=\ln(2)/2$ and $v_0=1$, it is simple to verify that the point $A=(\ln(2)/2,1)$ is above the critical amplitude curve represented in Fig.\ \ref{fig:vc} and, therefore, there is at least one bound state associated to the corresponding Sturm-Liouville Problem \ref{schroProb}, see Remark \ref{finitenessOfBoundSolutionsRemark}. Specifying these values in the formulas \eqref{v+-}, \eqref{nCondBounded}, and \eqref{anbn}, it can be found that $v_{+}=2$, $v_{-}=1/2$, $\alpha(v_{-})=(\sqrt{11}-\sqrt{3})/{\sqrt{2}}$, $b_0=(\sqrt{11/2}-1)/2$, and $a_0=3/(8 b_0)$. Accordingly to condition \eqref{nCondBounded}, there is only one valid index associated to a bound solution, which corresponds to $n=0$. From Eq. \eqref{e_n} and \eqref{eigenfunctionsSchroMF} we obtain $\varepsilon_0=(2 \sqrt{22}-5)/9$, and its corresponding eigenfunction 
\begin{equation} \label{ej1-psi1-r1}
	\psi_{0}(z)=\mathcal{N}_0 e^{-a_0 z} \sech^{b_0}(z), 
\end{equation}
respectively, where, recall, $\mathcal{N}_0$ is a normalizing factor.

Next, we deal with the continuous spectrum, $\varSigma$, by making use of Theorems \ref{theoUnbound} and \ref{theoUnbound1}. To begin with, since $\alpha(v_-) \notin \mathbb{Z}$, $\varepsilon=v_-=1/2 \notin \varSigma$ which means that the continuous spectrum satisfies $\varSigma = (1/2, \infty)$. Now, for its associated region \ref{R2}, this is, for $\varepsilon \in (1/2,2]$, the unbound solution is given by
\begin{equation} \label{ej1-psi1-r2}
	\psi_{1}(z;\varepsilon)=e^{-a(\varepsilon) z} \sech^{b(\varepsilon)}(z) F\left(b(\varepsilon)-b_0,b(\varepsilon)+b_0+1,\sqrt{2-\varepsilon}+1;\frac{1-\tanh z}{2}\right), 
\end{equation}
where $a(\varepsilon)=(\sqrt{2-\varepsilon}-i \sqrt{\varepsilon-1/2})/2$ and $b(\varepsilon)=a^\star(\varepsilon)$. Of course, the above function is not of integrable square but it is bounded. Finally, for its associated region \ref{R3}, this is, for $\varepsilon \in (2,+\infty)$, the unbound solutions are given by linear combinations of the following eigenfunction
\begin{equation*} 
	\psi_{1}(z;\varepsilon)=e^{-a(\varepsilon) z} \sech^{b(\varepsilon)}(z) F\left(b(\varepsilon)-b_0,b(\varepsilon)+b_0+1,i\sqrt{\varepsilon-2}+1;\frac{1-\tanh z}{2}\right), 
\end{equation*}
and its complex conjugate (see Remark \ref{complexConjugateRemark}), where $a(\varepsilon)=i\,(\sqrt{\varepsilon-2}- \sqrt{\varepsilon-1/2})/2$ and $b(\varepsilon)=i\,(\sqrt{\varepsilon-2}+ \sqrt{\varepsilon-1/2})/2$. 

\subsection{No bound states}
For our second example, we take $v_0=2/3$ and keep the same value for $\mu=\ln(2)/2$. This choice corresponds to point B in Fig.\ \ref{fig:vc}. Again, by employing \eqref{v+-}, \eqref{nCondBounded}, it can be found that 
$v_{-}=1/3$, $v_{+}=4/3$ and $N(\mu,v_0)=0=l$, according to previous notation. 
This means that, in this case, there are no bound states and the discrete spectrum is empty, according to Remark \ref{finitenessOfBoundSolutionsRemark}. 

However, this also implies (see Theorem \ref{theoUnbound} and \ref{theoUnbound1}) that the continuous spectrum $\varSigma = [1/3,\infty)$, this is, the limit state $\varepsilon=v_-$ is an actual eigenvalue. Therefore, the eigenfunction 
\begin{equation}\label{ej2-psi1-a}
	\psi_{0}(z) \propto e^{-z/2} \sech^{1/2}(z), 
\end{equation}
corresponds to the limit state $\varepsilon=v_-$, the following bounded function 
\begin{equation}\label{ej2-psi1-b}
	\psi_{1}(z;\varepsilon) \propto e^{-a(\varepsilon) z} \sech^{b(\varepsilon)}(z) F\left(b(\varepsilon)-1/2,b(\varepsilon)+3/2,\sqrt{4/3-\varepsilon}+1;\frac{1-\tanh z}{2}\right), 
\end{equation}
where $a(\varepsilon)=(\sqrt{4/3-\varepsilon}-i \sqrt{\varepsilon-1/3})/2$, and $b(\varepsilon)=a^\star(\varepsilon)$, corresponds to region \ref{R2}, this is, for $\varepsilon \in (1/3,4/3]$; and the free states region, this is, eigenenergies $\varepsilon \in (4/3,\infty)$, is represented through linear combinations of 
\begin{equation}\label{ej2-psi1-c}
	\psi_{1}(z;\varepsilon)=e^{-a(\varepsilon) z} \sech^{b(\varepsilon)}(z) F\left(b(\varepsilon)-1/2,b(\varepsilon)+3/2,i\sqrt{\varepsilon-4/3}+1;\frac{1-\tanh z}{2}\right), 
\end{equation} 
and its complex conjugate, where $a(\varepsilon)=i (\sqrt{\varepsilon-4/3}-\sqrt{\varepsilon-1/3})/2$, and $b(\varepsilon)=i (\sqrt{\varepsilon-4/3}+\sqrt{\varepsilon-1/3})/2$.

\subsection{Kink stability of $\varphi^{2p+2}$ field theory}
The static kink solution, $\varphi^{st}_{p}(x)$, of Eq.\  (\ref{eq:nlkg}) for the potential 
\begin{equation}
	\label{eq:pot}
	U(\varphi)=\frac{\varphi^2 (\varphi^p-2 \varphi_1)^2}{4 p^2},
\end{equation}
where $p \in \mathbb{N}, (p \ge 2)$, and $\varphi_1>0$, is given by 
\cite{casahorran:1991,saxena:2019}
\begin{equation}
	\label{eq:kink}
	\varphi_{p}^{st}(x)=\varphi_1^{1/p} \left(1+\tanh\left[\frac{\varphi_1 x}{\sqrt{2}}\right]\right)^{1/p}.
\end{equation}
The moving soliton can be obtained from $\varphi_{p}^{st}(x)$ by using the Lorentz transformation.  
The Lorentz invariance of Eq.\  (\ref{eq:nlkg}) also allows the kink (antikink) stability analysis to be reduced to considering a perturbation $\Psi(x,t)$ around the static kink (antikink). Therefore, Eq.\ (\ref{eq:nlkg})  is linearized around its 
static solution, that is, the function \cite{parmentier:1967,scott:1969} 
	$\varphi(x,t)= \nvarphi(x)+(c_1 e^{i\omega t}+c_2 e^{-i\omega t})\,\psi(x)$, 
where $\omega^2\in\mathbb{R}$,  \cite{parmentier:1967,scott:1969,saxena:2019} 
is introduced in Eq.\ (\ref{eq:nlkg}). The  
complex constants $c_1$ and $c_2$ used in the ansatz  are chosen such that $\Psi(x,t)\in\mathbb{R}$. As a consequence, 
$\psi(z)$ verifies the following  Sturm-Liouville problem
\begin{equation}
	\label{eq:sg8a}
	\psi_{zz}(z)+\left[\Lambda^2-V(z)\right]\psi(z)=0, 
\end{equation}
where $z=\varphi_1 x /\sqrt{2}$,  $\Lambda^2=2(\omega^2-\omega_{ph}^2)/\varphi_1^2$, being 
\begin{equation}
	\label{eq:wph}
	\omega_{ph}^2=\min\left\{\lim_{x \to - \infty} U''[\nvarphi(x)], 
	\lim_{x \to + \infty} U''[\nvarphi(x)] \right\}=\frac{2\varphi_1^2}{p^2},
\end{equation}
and the function $V(z)=2\,(U''[\nvarphi(z)]-\omega_{ph}^2)/\varphi_1^2$ satisfies that
\begin{equation}
	\label{eq:Vx}
	V(z)=-\frac{3(p+1)}{p^2} +\frac{2(p^2-1)}{p^2} \tanh(z)+\frac{(2p+1)(p+1)}{p^2}\tanh^2(z),
\end{equation}
which is the well-known Rosen-Morse potential \cite{rosen:1932}. 
In particular, for $p=1$, Eq.\ \eqref{eq:Vx} reduces to the symmetric P\"oschl-Teller potential \cite{poschl:1933}. 
Equation \eqref{eq:sg8a} can be transformed into Eq.\ 
\eqref{schrEqMFz} by denoting 
\begin{equation} \label{eq:para} 
	v_{0} = \dfrac{3(p+1)(p+2)}{p (2p+1)}, \, \,
	\tanh \mu= \dfrac{p-1}{2p+1}, \, \, \varepsilon= \Lambda^2+\dfrac{v_{0}(p+2)}{3p}, \, \, v(z)=V(z)+\dfrac{v_0 (p+2)}{3p}.
\end{equation} 
Therefore, from Eqs.\  \eqref{nCondBounded} and \eqref{eq:para} it can be shown that 
$v_{-}=(p+2)^2 (p+1)/(p^2 (2p+1))$, $v_{+}=9 (p+1)/(2p+1)$, and  
$\Lambda^2(\varepsilon) \equiv \varepsilon-v_-$. It is straightforward to verify that, here, the condition $v_0> e^{2\mu}\tanh \mu$ is satisfied for all values of $p$, see the orange circles  represented in Fig.\ \ref{fig:vc} for some values of $p$. Only the limit value of $(\mu,v_0)$ when $p \to +\infty$ lies on the curve (see the red square in Fig.\ \ref{fig:vc}). Therefore, there always is at least one bound state for every $p \ge 2$. 

Indeed, from Eq. \eqref{nCondBounded} and by denoting $N(\mu,v_0) = N(p)$, it can be found that
$$
N(p)=\dfrac{p+1 -\sqrt{p^2-1}}{p},
$$
meaning $0<N(p)<1$ for all $p \ge 2$. This implies that there always is only one bound state, and the solutions of Eqs.\ \eqref{anbn}  and \eqref{nCondBounded} are 
$n=0$, $a_0=(p-1)/p$, and $b_0=(p+1)/p$. From Eqs.   \eqref{eigenfunctionsSchroMF} and \eqref{e_n}, we find $\varepsilon_0=(p+5)/(2 p +1)$, and its corresponding eigenfunction 
\begin{equation*} 
	\psi_{0}(z)=\mathcal{N}_0 e^{-(p-1) z/p} \sech^{(p+1)/p}(z), 
\end{equation*}
respectively. Notice that, in this case, from Eq.\ \eqref{eq:para}, $\Lambda^{2}(\varepsilon_0)=-4/p^2$, implying that $\omega(\varepsilon_0)=0$. This frequency corresponds to the well-known Goldstone mode \cite{goldstone:1975} meaning $\psi_0(z) \propto d\varphi_{p}^{st}/dz$. 

In addition, given that $0<N(p)<1$, the state $\varepsilon=v_{-}$ does not belong to the continuum. Hence, $\varepsilon \in (v_{-},v_{+}]$ and $\varepsilon \in(v_{+},+\infty)$ define the associated \ref{R2} and \ref{R3} regions, 
respectively; and theirs corresponding eigenfunctions can be obtained in a similar way as shown in the previous examples. In this way, for the remaining frequencies, $\omega(\varepsilon)$, associated to the continuous spectrum, $\varSigma = (v_-,\infty)$, given that $\Lambda^2(\varepsilon)>\Lambda^2(v_-)=0$, they satisfy that 
\begin{equation*}
\omega^2(\varepsilon)>\omega^2_{ph}=\frac{2\varphi^2_1}{p^2}>0,
\end{equation*}
meaning the kinks are stable for every value of $p \ge 2$, $p \in \mathbb{N}$.

The case $p=2$ was studied in Ref. \cite{lohe:1979}, where 
states $\varepsilon=v_{-}$ and $\varepsilon=v_{+}$ were considered as part of regions \ref{R2} and \ref{R3}, respectively. However, through this section, we have shown that the limit state $\varepsilon=v_{-}$ cannot belong to region \ref{R2}, since, its associated eigenfunctions are unbounded accordingly to Theorem \ref{theoUnbound1}; and, also, that the limit state $\varepsilon=v_{+}$ cannot belong to region \ref{R3}. Indeed, it can be verified that the corresponding eigenfunction defined in Eq. (31) of \cite{lohe:1979} diverges as $x \to -\infty$ whenever $\varepsilon=v_{-}$. Moreover, the eigenfunctions defined by Eqs.\ (33a)-(34a) in Ref. \cite{lohe:1979} lose their oscillatory character as $x \to \infty$ and are linearly dependent whenever $\varepsilon=v_{+}$, but, since they are bounded, they should be considered as part of region \ref{R2}.

 \section{Conclusions } \label{sec5}

 This work investigates the solutions $(\varepsilon,\psi(z))$ of the  time-independent Schr\"odinger equation for the Rosen-Morse type potential.  This equation is first transformed into a generalized hypergeometric equation or GHE by using the universal change of variable $u=-\tanh(z)$.  
 By means of the Nikiforov-Uvarov method, its solution is written as $\psi(u)=\phi(u)\,y(u)$, where $y(u)$ satisfies another GHE which is simpler than the one $\psi(u)$ solves. 
 To achieve this simplification, $\phi(u)$ is forced to fulfill the  
  first-order linear differential Eq. \eqref{phi} for a convenient $\pi(u)$, which is easily obtained once all the coefficients and polynomials related to the HDE are determined.  
  
  Now, since the Rosen-Morse potential is asymmetric for $\mu > 0$, we distinguish tree different regions of eigenenergies, $\varepsilon$, by employing the asymptotic values of the potential, $v_{\pm}=v_0 e^{\pm 2 \mu}$. In this way, the main results of 
  the current investigation are formulated and proven in theorems \ref{theoBounded}   
  and \ref{theoUnbound}, which characterize, respectively, the discrete and the continuous spectra. In the former case, 
  the proof of the theorem follows from the Nikiforov-Uvarov theorem \ref{theoA1}, which allows us to obtain the set of 
  $n_b$ bound solutions, $\psi_{n}(z)$, defined in Eq. \eqref{eigenfunctionsSchroMF}, where their corresponding $\varepsilon_n \in (0,v_{-})$ satisfies relation \eqref{eigenvalueEqMF}. In the later case, their associated eigenfunctions, 
  $\psi(z)$, are no longer necessarily square-integrable functions, but should remain bounded. As a consequence, for every  
 $\varepsilon \in (v_{-},v_+)$, this is, for \ref{R2}, their unbound solution is only represented by the one eigenfunction defined through Eq.\ \eqref{solRefStates}, as the second linearly independent solution of the Gauss' HDE is disregarded for being unbounded. On the contrary, for every $\varepsilon \in (v_{+},\infty)$, this is, for region \ref{R3}, both linearly independent solutions, $\psi_1(z;\varepsilon)$ and $\psi_2(z;\varepsilon)=\psi_1^{\star}(z;\varepsilon)$ given by \eqref{solUnboundEq}, of the Gauss' HDE satisfy the boundedness condition which ultimately implies that their unbound solutions are represented through Eq.\ \eqref{solFreeStates}. Theorem \ref{theoUnbound} is completed by Theorem \ref{theoUnbound1} where the limit energy states $\varepsilon=v_{\pm}$ are considered. In this theorem, it is proven that $\varepsilon=v_{+}$ is an energy state corresponding to \ref{R2} and therefore, its  unbound solution is given by Eq.\  \eqref{psivplus}. Moreover, it is shown that the state $\varepsilon=v_-$ belongs to the continuous spectrum only for certain values of the pair $(\mu,v_0)$. Specifically, it only belongs to region \ref{R2} whenever $N(\mu,v_0) = l$ for some $l=0,1,2,\dots$
 
  In the first two examples of Section \ref{sec3}, it is discussed how the two control parameters of the Rosen-Morse potential,   $\mu$ and $v_0$, modify the shape of the associated spectrum. Firstly, by modifying the number of bound states, $n_b$, through inequality \eqref{nCondBounded}. Secondly, by directly changing the computation of the bound solutions, $(\varepsilon_n,\psi_n(z))$, through Eqs.\  \eqref{e_n}, \eqref{eigenfunctionsSchroMF} and Eq. \eqref{anbn}; and of the unbound solutions, $(\varepsilon, \psi_{\varepsilon}(z))$, through Eqs.\ \eqref{R2ab}, \eqref{R3ab} and \eqref{abgGausshdeMF}. In addition, for the second example, the special condition $N(\mu,v_0)=0$ is satisfied which implies that $\varepsilon=v_-$ is actually part of the continuous spectrum.
  
  As the main application of the results obtained in Section \ref{sec2}, it is justified the stability of the static kinks, $\varphi_{p}^{st}(x)$, of the non-linear Klein-Gordon equations with $\varphi^{2p+2}$ ($p \ge 2$) type potentials. We also show that for every $p \ge 2$, the spectrum of the corresponding Sturm-Liouville problem contains only one bound state, the so-called Goldstone mode, with zero frequency and with eigenfunction $\psi_0(x) \propto d\varphi_{p}^{st}(x)/dx$, which means that the system has no internal modes. This result answers the question about the existence of internal modes and its dependence on $p$  raised  by Saxena, Christov and Khare in Ref. \cite{saxena:2019}.  Moreover, we exactly define limits of the continuous spectrum, which implies that the continuum for the frequencies is given by $\omega \in (\omega_{ph},+\infty)$, where, notice, the frequency $\omega=\omega_{ph}$ is excluded. Exception occurs for $p=1$ ($\mu=0$). In this case, the Rosen-Morse potential becomes into the P\"osch-Teller potential, the discrete spectrum contains not only the Goldstone mode, but also one internal mode, and in addition $\omega=\omega_{ph}$ belongs to the continuous limit (see \cite{raban:2022}).
  
Finally, the orthogonality and completeness relations satisfied by the set of eigenfunctions, are  given in theorem \ref{theo2.4}. 
Knowledge of these relationships will be useful in describing the dynamics of kinks under external perturbations  or interacting with antikinks.

\appendix

\section{Some results related to the hypergeometric type functions} \label{apen}

In this Appendix we will enumerate some relevant results needed for a better understanding of the
main results of the present work. We start with the following theorem \cite[page 67]{nu:1988}.

\begin{theorem}[Nikiforov-Uvarov] \label{theoA1}
Let the hypergeometric differential equation
\begin{equation}
\label{AHDE} \sigma(u) y''(u) + \tau(u) y'(u) +\lambda y(u)=0, \quad u \in  (a,b).
\end{equation}
Let $\rho(u) \in L^{\infty}(a,b)$ and satisfy condition 
\begin{equation}
\sigma(u)\rho(u)u^k\Bigr|_{x=a} =    \sigma(u)\rho(u)u^k\Bigr|_{u=b}=0,\quad \forall k \in \mathbb{N} \cup \{0\},
\label{orthoPol}
\end{equation}
where $\rho(u)$ is any solution of the Pearson equation
\begin{equation}
(\sigma(u)\rho(u))'=\tau(u)\rho(u).
\label{rho}
\end{equation}
Then, non-trivial solutions of \eqref{AHDE}, $y_{\lambda}(u)$, satisfying that $y_{\lambda}(u)\sqrt{\rho(u)} \in L^{2}(a,b) \cap L^{\infty}(a,b)$ 
exist only when
\begin{equation}\label{eq-lambda}
\lambda = \lambda_n = -n\tau'(u)-\frac{n(n-1)}{2}\sigma''(u), \quad n \in \mathbb{N} \cup \{0\}
\end{equation}
and have the form $(y_{\lambda}(u)=y_{\lambda_n}(u)=y_{n}(u))$
\begin{equation}\label{for-rod}
y_n(u)=\frac{B_n}{\rho(u)}\left[\sigma^n(u)\rho(u)\right]^{(n)}, 
\end{equation}
where $B_n$ is a constant.
Therefore, they are the classical orthogonal polynomials associated to the weight function $\rho(u)$, i.e., they
satisfy the orthogonality relation 
\begin{equation}\label{rel-ort-pol}
\int_{a}^{b} y_n(u)y_m(u)\rho(u)\,du=\delta_{n,m} d_n^2,
\end{equation}
where $\delta_{n,m}$ is the Kronecker delta symbol and $d_n^2$ is the squared norm of the 
polynomials. 
\label{eigenTheo}
\end{theorem} 
The following Proposition \cite[page 67]{nu:1988} will be also useful.

\begin{prop}
\label{easesChoiceProp}
Let $\rho(u)$ satisfy condition \eqref{orthoPol} for the classical orthogonal polynomials on $(a,b)$. Then, $\tau(u)$ has to vanish at some point of $(a,b)$ and $\tau'(u)<0$.
\end{prop}
The properties of the Jacobi polynomials are summarized in the following proposition \cite[\S5]{nu:1988}.
\begin{prop}\label{jac-pol}
The Jacobi polynomials $P_n^{(\alpha,\beta)}(u)$ are the polynomial solutions of the
HDE \eqref{AHDE} with 
\begin{equation*}
\sigma(u)=1-u^{2}, \; \; \tau(u)=-\left(\alpha+\beta+2\right)u+\beta-\alpha,\quad \lambda_n= n\left(n+\alpha+\beta+1\right).
\label{jacobiPar}
\end{equation*}
They satisfy the orthogonality relation \eqref{rel-ort-pol} on $[-1,1]$ with respect to the weight function $\rho(u)=\left(1-u\right)^{\alpha}\left(1+u\right)^{\beta}$,
with for  $\alpha>-1, \; \beta>-1$, and
$$
d^2_n=\frac{2^{\alpha+\beta+1} \Gamma(n+\alpha+1) \Gamma(n+\beta+1)}{n!(2n+\alpha+\beta+1)\Gamma(n+\alpha+\beta+1)}.
$$
They can be calculated by the Rodrigues formula \eqref{for-rod} with $B_n=\frac{(-1)^n}{2^nn!}$, and 
can be written in terms of the Gauss's terminating hypergeometric series \eqref{hypergeometricSeries}
\begin{equation}\label{jac-hyp}
\begin{split}
\displaystyle
P_n^{(\alpha,\beta)}(u)= & \frac{(\alpha+1)_n}{n!} \displaystyle\,_2
F _1\bigg(\displaystyle \begin{array}{c}{-n,n+\alpha+\beta+1 }\\{ \alpha+1}\end{array}\bigg|
\frac{1-u}{2} \bigg).
\end{split}
\end{equation}
\end{prop}
Finally, for the hypergeometric functions, we have the following result \cite[\S20-21]{nu:1988}.
\begin{prop}
\label{propGaussHypergeometricFunction}
 Gauss's hypergeometric function, $F(\alpha,\beta,\gamma;u)$ is defined by the power series 
\begin{equation}
F(\alpha,\beta,\gamma;u)=\sum_{n=0}^{\infty}\frac{(\alpha)_n (\beta)_n}{(\gamma)_n n!}u^n, \quad \gamma \neq -k,\,\, \forall k \in \mathbb{N} \cup \{0\},
\label{hypergeometricSeries}
\end{equation}
which is analytic function for  $\left|u\right|<1$. Here $(a)_k$ denotes the Pochhammer symbol 
\begin{equation}\label{poch}
(a)_0=1 ,\hspace{.25cm} (a)_k=a(a+1)(a+2) \cdots(a+k-1),\quad
k=1,2,3,\dots \,\,\, , 
\end{equation}
The function $F(\alpha,\beta,\gamma;u)$ is a solution of the so-called Gauss' differential equation 
\begin{equation}
u\left(1-u\right)y''(u)+\left[\gamma-\left(\alpha+\beta+1\right)u\right]y'(u)-\alpha \beta y(u)=0.
\label{gaussEq}
\end{equation}
\end{prop}

\section{Computing the    normalizing factor $\mathcal{N}_n$ in Theorem \ref{theoBounded}}
\label{comp-N_n}
Here, we will show how the normalizing factor $\mathcal{N}_n$ in Theorem \ref{theoBounded} can be computed. 
From Eq. \eqref{eigenfunctionsSchroMF} it follows that 
\begin{equation}\label{int-norm}
\begin{split}
1 & =\int_\mathbb{R} | \psi_{\varepsilon_n}(z)|^2dz=\mathcal{N}_n^2 
 \int_{-1}^{1}(1-u)^{b_n-a_n-1}(1+u)^{b_n+a_n-1}  [P_{n}^{(b_n-a_n,b_n+a_n)}\left(u\right)]^2 du \\
 & = \mathcal{N}_n^2  2^{2b_n-1}\int_{0}^{1}(1-\zeta)^{b_n-a_n-1}\zeta^{b_n+a_n-1}  [P_{n}^{(b_n-a_n,b_n+a_n)}\left(2\zeta-1\right)]^2 d\zeta.
\end{split}
\end{equation}
For $n=0$ the above integral gives
$$
\mathcal{N}_0^2 =  2^{1-2b_0} [B(b_0-a_0,b_0+a_0)]^{-1},
$$
where $B$ denotes de Beta function of Euler \cite[\S5.12]{olver:2010}. 

Whenever $n\geq1$, a general analytical formula can be obtained by using the so-called linearization coefficients  
as well as the explicit expression of the Jacobi polynomials. Since the analytical formulas are quite cumbersome 
(see e.g. \cite{chaggara:2010,rahman:1981}), we will show a simpler way for obtaining them. In order to do so, we can use 
the following explicit expression for the Jacobi polynomials \cite[Eq. 18.18.26]{olver:2010}
\begin{equation} \nonumber
P_{n}^{(b_n-a_n,b_n+a_n)}\left(2\zeta-1\right)=
\sum_{\ell=0}^{n}c_{\ell,n}\zeta^\ell,   \quad c_{\ell,n}=\frac{(n+2b_n+1)_\ell (-b_n-a_n-n)_{n-\ell}}{\ell!(n-\ell)!},
\end{equation}
and compute 
\begin{equation} \nonumber
[P_{n}^{(b_n-a_n,b_n+a_n)}\left(2\zeta-1\right)]^2=\sum_{k=0}^{2n}  \pi_{k,n}  \zeta^k,
\quad 
\pi_{k,n}=  \sum_{\ell=0}^k c_{\ell,n} c_{k-\ell,n}.
\end{equation}

Thus, Eq. \eqref{int-norm}
leads to 
\begin{equation} \nonumber
\mathcal{N}_n^2  2^{2b_n-1}\sum_{k=0}^{2n}   \pi_{n,k} B(b_n-a_n,b_n+a_n+k)=1,
\end{equation}
from where $\mathcal{N}_n^2$ can be easily found. 

\section*{Acknowledgments}
We thanks Siannah P. Rivas for drawing the Ref. \cite{casahorran:1991} to our attention. 
R.A.N. was partially supported by PID2021-124332NB-C21 (FEDER(EU)/Ministerio de Ciencia e 
Innovación-Agencia Estatal de Investigación) and FQM-262 (Junta de Andalucía).
N.R.Q. was partially supported by the Spanish projects PID2020-113390GB-I00 (MICIN), PY20$\_$00082 
(Junta de Andalucia),  A-FQM-52-UGR20 (ERDF-University of Granada), and the Andalusian research group FQM-207.




\begin{thebibliography}{10}
	
	\bibitem{bogdan:1990}
	{ M. M. Bogdan and A. M. Kosevich and V. P. Voronov}.
	\newblock {Generation of the internal oscillation of soliton in a
		one-dimensional non-integrable system}.
	\newblock In {V. G. Makhankov and V. K. Fedyanin and O. K. Pashaev}, editor,
	{\em Proceedings of IVth International Workshop ``Solitons and
		Applications''}, pages 397--401. World Scientif\/ic, Part~IV, Singapore,
	1990.
	
	\bibitem{rezaei:2008}
	Rezaei A. and H.~Motavalli.
	\newblock {Exact Solutions of the Klein-Gordon Equation for the Rosen-Morse
		Type Potentials via Nikiforov-Uvarov Method}.
	\newblock {\em Mod. Phys. Lett. A}, 23:3005, 2008.
	
	\bibitem{aubry:1976}
	S.~Aubry.
	\newblock {A unified approach to the interpretation of displacive and
		order-disorder systems. II. Displacive systems}.
	\newblock {\em J. Chem. Phys.}, 64:3392, 1976.
	
	\bibitem{aubry:1976a}
	S.~Aubry.
	\newblock A new interpretation of the dynamical behaviour of a displacive
	model.
	\newblock {\em Ferroelectrics}, 12:263, 1976.
	
	\bibitem{barashenkov:2019}
	I.~V. Barashenkov.
	\newblock {The Continuing Story of the Wobbling Kink}.
	\newblock In Panayotis~G. Kevrekidis and Jes{\'u}s Cuevas-Maraver, editors,
	{\em A Dynamical Perspective on the $\phi ^4$ Model: Past, Present and
		Future}, pages 187--212. Springer International Publishing, Cham, 2019.
	
	\bibitem{barashenkov:2009}
	I.~V. Barashenkov and O.~F. Oxtoby.
	\newblock {Wobbling kinks in ${\ensuremath{\phi}}^{4}$ theory}.
	\newblock {\em Phys. Rev. E}, 80:026608, 2009.
	
	\bibitem{bishop:1980}
	A.~R. Bishop, J.~A. Krumhansl, and S.~E. Trullinger.
	\newblock {Solitons in condensed matter: A paradigm}.
	\newblock {\em Phys. D: Nonlinear Phenom.}, 1:1, 1980.
	
	\bibitem{campbell:2019}
	D.~K. Campbell.
	\newblock {Historical Overview of the $\phi ^4$ Model}.
	\newblock In Panayotis~G. Kevrekidis and Jes{\'u}s Cuevas-Maraver, editors,
	{\em {A Dynamical Perspective on the $\phi^4$ Model: Past, Present and
			Future}}, pages 1--22. Springer International Publishing, Cham, 2019.
	
	\bibitem{campbell:1983}
	D.~K. Campbell, J.~F. Schonfeld, and C.~A. Wingate.
	\newblock {Resonance structure in kink-antikink interactions in $\varphi^4$
		theory}.
	\newblock {\em Phys. D: Nonlinear Phenomena}, 9:32, 1983.
	
	\bibitem{casahorran:1991}
	J.~Casahorran.
	\newblock Solitary waves and polynomial potentials.
	\newblock {\em Phys. Lett. A}, 153:199, 1991.
	
	\bibitem{chaggara:2010}
	H.~Chaggara and W.~Koepf.
	\newblock {On linearization coefficients of Jacobi polynomials}.
	\newblock {\em Appl. Math. Lett.}, 23:609, 2010.
	
	\bibitem{condat:1983}
	C.~A. Condat, R.~A. Guyer, and M.~D. Miller.
	\newblock {Double sine-Gordon chain}.
	\newblock {\em Phys. Rev. B}, 27:474, 1983.
	
	\bibitem{cooper:1995}
	F.~Cooper, A.~Khare, and U.~Sukhatme.
	\newblock Supersymmetry and quantum mechanics.
	\newblock {\em Phys. Rep.}, 251:267, 1995.
	
	\bibitem{dashen:1974}
	R.~F. Dashen, B.~Hasslacher, and A.~Neveu.
	\newblock {Nonperturbative methods and extended-hadron models in field theory.
		II. Two-dimensional models and extended hadrons}.
	\newblock {\em Phys. Rev. D}, 10:4130, 1974.
	
	\bibitem{dauxois:2006}
	Th. Dauxois and M.~Peyrard.
	\newblock {\em {Physics of Solitons}}.
	\newblock Cambridge University Press, Cambridge, 2006.
	
	\bibitem{dorey:2011}
	P.~Dorey, K.~Mersh, T.~Romanczukiewicz, and Y.~Shnir.
	\newblock {Kink-Antikink Collisions in the ${\ensuremath{\phi}}^{6}$ Model}.
	\newblock {\em Phys. Rev. Lett.}, 107:091602, 2011.
	
	\bibitem{eckart:1930}
	C.~Eckart.
	\newblock {The penetration of a potential barrier by electrons}.
	\newblock {\em Phys. Rev.}, 35:1303, 1930.
	
	\bibitem{epstein:1930}
	P.~S. Epstein.
	\newblock {Reflexion of waves in an inhomogeneous absorbing medium}.
	\newblock {\em {Proc. Natl. Acad. Sci. USA.}}, 16:627, 1930.
	
	\bibitem{olver:2010}
	{F. W. J. Olver, D. W. Lozier, R. F. Boisvert, and C. W. Clark May}, editor.
	\newblock {\em {NIST Handbook of Mathematical Functions Paperback and CD-ROM}}.
	\newblock {Cambridge University Press}, New York, 2010.
	
	\bibitem{forgacs:2008}
	P.~Forg\'acs, \'A. Luk\'acs, and T.~Roma\'{n}czukiewicz.
	\newblock Negative radiation pressure exerted on kinks.
	\newblock {\em Phys. Rev. D}, 77:125012, 2008.
	
	\bibitem{fullin:1978}
	W.~C. Fullin.
	\newblock {One-dimensional field theories with odd-power self-interactions}.
	\newblock {\em Phys. Rev. D}, 18:1095, 1978.
	
	\bibitem{getmanov:1976}
	B.~S. Getmanov.
	\newblock {Soliton Bound States in the $\phi^4$ in Two-Dimensions Field
		Theory}.
	\newblock {\em Pisma Zh. Eksp. Teor. Fiz.}, 24:323, 1976.
	
	\bibitem{ginzburg:2009}
	V.~L. Ginzburg and L.~D. Landau.
	\newblock {\em On the Theory of Superconductivity}, pages 113--137.
	\newblock Springer Berlin Heidelberg, Berlin, Heidelberg, 2009.
	
	\bibitem{goldstone:1975}
	J.~Goldstone and R.~Jackiw.
	\newblock {Quantization of nonlinear waves}.
	\newblock {\em {Phys. Rev. D}}, 11:1486, 1975.
	
	\bibitem{gor:2022}
	G.~Gordillo-Nuñez.
	\newblock {\em {Study of the solutions of certain class of Schr\"odinger
			equations}}.
	\newblock Degree in Physics and Mathematics, Facultad de Matemáticas.
	Universidad de Sevilla, 2022. Avaliable online at \url{https://idus.us.es/handle/11441/144162}
	
	\bibitem{kevrekidis:2001}
	P.~G. Kevrekidis.
	\newblock {Integrability revisited: a necessary condition}.
	\newblock {\em Phys. Lett. A}, 285:383, 2001.
	
	\bibitem{khare:1979}
	A.~Khare.
	\newblock A. static finite energy solutions of a classical field theory with
	positive mass-square.
	\newblock {\em Lett. Math. Phys.}, 3:475, 1979.
	
	\bibitem{khare:2014}
	A.~Khare, I.~C. Christov, and A.~Saxena.
	\newblock {Successive phase transitions and kink solutions in
		${\ensuremath{\phi}}^{8}$, ${\ensuremath{\phi}}^{10}$, and
		${\ensuremath{\phi}}^{12}$ field theories}.
	\newblock {\em Phys. Rev. E}, 90:023208, 2014.
	
	\bibitem{kiselev:1998}
	V.~G. Kiselev and Ya.~M. Shnir.
	\newblock Forced topological nontrivial field configurations.
	\newblock {\em Phys. Rev. D}, 57:5174, 1998.
	
	\bibitem{koehler:1975}
	T.~R. Koehler, A.~R. Bishop, J.~A. Krumhansl, and J.~R Schrieffer.
	\newblock Molecular dynamics simulation of a model for (one-dimensional)
	structural phase transitions.
	\newblock {\em Solid State Commun.}, 17:1515, 1975.
	
	\bibitem{krumhansl:1975}
	J.~A. Krumhansl and J.~R. Schrieffer.
	\newblock {Dynamics and statistical mechanics of a one-dimensional model
		Hamiltonian for structural phase transitions}.
	\newblock {\em Phys. Rev. B}, 11:3535, 1975.
	
	\bibitem{takh:2008}
	{L. A. Takhtajan}.
	\newblock {\em {Quantum mechanics for mathematicians.}}
	\newblock Graduate Studies in Mathematics, 95. American Mathematical Society,
	Providence, R.I., 2008.
	
	\bibitem{landau:1950}
	L.~D. Landau and V.~L. Ginzburg.
	\newblock {On the theory of superconductivity}.
	\newblock {\em Zh. Eksp. Teor. Fiz.}, 20:1064, 1950.
	
	\bibitem{lohe:1979}
	M.~A. Lohe.
	\newblock {Soliton structures in $P{(\ensuremath{\varphi})}_{2}$}.
	\newblock {\em Phys. Rev. D}, 20:3120, 1979.
	
	\bibitem{makhankov:1978}
	V.~G. Makhankov.
	\newblock {Dynamics of classical solitons (in non-integrable systems)}.
	\newblock {\em {Phys. Rep.}}, 35:1, 1978.
	
	\bibitem{morse2:1953}
	P.~M. Morse and H.~Feshbach.
	\newblock {\em {Methods of theoretical Physics}}, volume~2.
	\newblock McGraw-Hill, New York, 1953.
	
	\bibitem{morse:1953}
	P.~M. Morse and H.~Feshbach.
	\newblock {\em {Methods of theoretical Physics}}, volume~1.
	\newblock McGraw-Hill, New York, 1953.
	
	\bibitem{nu:1988}
	A.~F. Nikiforov and V.~B. Uvarov.
	\newblock {\em {Special Functions of Mathematical Physics}}.
	\newblock {Birkh{\"a}user Verlag}, Basilea, 1988.
	
	\bibitem{parmentier:1967}
	R.~D. Parmentier.
	\newblock {Stability analysis of neuristor waveforms}.
	\newblock {\em Proc. IEEE}, 55:1498, 1967.
	
	\bibitem{peyrard:1983}
	M.~Peyrard and D.~K. Campbell.
	\newblock {Kink-antikink interactions in a modified sine-Gordon model}.
	\newblock {\em Phys. D: Nonlinear Phenom.}, 9:33, 1983.
	
	\bibitem{poschl:1933}
	G.~P{\"o}schl and E.~Teller.
	\newblock {Bemerkungen zur Quantenmechanik des anharmonischen Oszillators}.
	\newblock {\em Z. Physik}, 83:143, 1933.
	
	\bibitem{strichartz:1994}
	{R. Strichartz}.
	\newblock {\em {A guide to distribution theory and Fourier transforms}}.
	\newblock Studies in Advanced Mathematics. CRC Press, Boca Raton, FL, 1994.
	
	\bibitem{rahman:1981}
	M.~Rahman.
	\newblock {A Non-Negative Representation of the Linearization Coefficients of
		the Product of Jacobi Polynomials}.
	\newblock {\em Canad. J. Math.}, 33:915, 1981.
	
	\bibitem{roman:2019}
	T.~Romańczukiewicz and Y.~Shnir.
	\newblock {\em Some Recent Developments on Kink Collisions and Related Topics},
	pages 23--49.
	\newblock Springer International Publishing, Cham, 2019.
	
	\bibitem{rosen:1932}
	N.~Rosen and P.~M. Morse.
	\newblock {On the Vibrations of Polyatomic Molecules}.
	\newblock {\em Phys. Rev.}, 42:210, 1932.
	
	\bibitem{rubinstein:1970}
	J.~Rubinstein.
	\newblock {Sine‐Gordon Equation}.
	\newblock {\em {J. Math. Phys.}}, 11:258, 1970.
	
	\bibitem{raban:2022}
	P.~Rában, R.~Alvarez-Nodarse, and N.~R. Quintero.
	\newblock {Stability of solitary wave of nonlinear Klein-Gordon equations}.
	\newblock {\em {J. of Phys. A: Math. Theor.}}, 55:465201, 2022.
	
	\bibitem{saxena:2019}
	A.~Saxena, I.~C. Christov, and A.~Khare.
	\newblock {Higher-Order Field Theories: $\phi^6$, $\phi^8$ and Beyond}.
	\newblock In Panayotis~G. Kevrekidis and Jes{\'u}s Cuevas-Maraver, editors,
	{\em {A Dynamical Perspective on the $\phi^4$ Model: Past, Present and
			Future}}, pages 253--279. Springer International Publishing, Cham, 2019.
	
	\bibitem{scott:1969}
	A.~C. Scott.
	\newblock {Waveform stability on a nonlinear Klein-Gordon equation}.
	\newblock {\em Proc. IEEE}, 57:1338, 1969.
	
	\bibitem{sugiyama:1979}
	T.~Sugiyama.
	\newblock {Kink-Antikink Collisions in the Two-Dimensional $\phi^4$ Model}.
	\newblock {\em {Prog. Theor. Phys.}}, 61:1550, 1979.
	
	\bibitem{sugiyama:1970}
	T.~Sugiyama.
	\newblock {Kink-Antikink Collisions in the Two-Dimensional $\varphi^4$ Model}.
	\newblock {\em Prog. Theor. Phys.}, 61:1550, 1979.
	
\end{thebibliography}

\end{document}